\documentclass[final]{IEEEtran}
\usepackage{citesort}
\usepackage{graphicx}
\ifCLASSINFOpdf
\else
\fi
\usepackage[cmex10]{amsmath}
\usepackage{amssymb}

\usepackage{german}
\usepackage{url}
\usepackage{color}
\usepackage{colordvi}
\usepackage{xspace}
\newcommand{\RED}[1]{\textcolor{red}{#1}}

\newtheorem{corollary}{Corollary}[section]
\newtheorem{lemma}{Lemma}[section]
\newtheorem{proposition}{Proposition}[section]
\newtheorem{remark}{Remark}
\newtheorem{theorem}{Theorem}

\renewcommand{\vec}{\operatorname{vec}}
\newcommand{\N}{\ensuremath{\mathbb N}{}}
\newcommand{\R}[1]{\ensuremath{\mathbb R}^{\,#1}{}}
\newcommand{\C}[1]{\ensuremath{\mathbb C}^{\,#1}{}}

\newcommand{\unity}{\ensuremath{{\rm 1 \negthickspace l}{}}}

\newcommand{\expt}[1]{\ensuremath{\langle #1 \rangle}{}}

\newcommand{\adr}{\operatorname{ad}}
\newcommand{\Adr}{\operatorname{Ad}}

\renewcommand{\vec}{\operatorname{vec}{}}
\newcommand{\diag}{\operatorname{diag}{}}

\newcommand{\iso}{\ensuremath{\overset{\rm iso}{=}}\xspace}

\newcommand{\ct}{\ensuremath{\cos(\theta)}}
\newcommand{\st}{\ensuremath{\sin(\theta)}}

\newcommand{\tr}{\operatorname{tr}}
\newcommand{\comm}[2]{\ensuremath{[#1,#2]}}

\newcommand{\herm}{\mathfrak{her}{}}
\newcommand{\pos}{\mathfrak{pos}{}}

\newcommand{\Reach}{\operatorname{Reach}{}}
\newcommand{\conv}{\operatorname{conv}{}}

\newcommand{\su}{\mathfrak{su}}

\newcommand{\so}{\mathfrak{so}}

\newcommand{\gl}{\mathfrak{gl}}
\newcommand{\e}{{\rm e}}

\newcommand{\rT}{{\rm T}}

\newcommand{\fe}{\mathfrak{e}}

\newcommand{\fg}{\mathfrak{g}}
\newcommand{\fw}{\mathfrak{w}}
\newcommand{\fc}{\mathfrak{c}}
\newcommand{\fk}{\mathfrak{k}}
\newcommand{\fh}{\mathfrak{h}}

\newcommand{\fp}{\mathfrak{p}}
\newcommand{\sym}{\mathfrak{sym}}
\newcommand{\fs}{\mathfrak{s}}

\newcommand{\bG}{\textbf{G}}
\newcommand{\bS}{\textbf{S}}
\newcommand{\bK}{\textbf{K}}
\newcommand{\bH}{\textbf{H}}
\newcommand{\bP}{\textbf{P}}
\newcommand{\bT}{\textbf{T}}
\newcommand{\cLu}{\mathcal{L}_u}

\newcommand{\sA}{\ensuremath{\sf{A}}\xspace}
\newcommand{\sB}{\ensuremath{\sf{B}}\xspace}

\begin{document}
\title{Illustrating the Geometry of Coherently Controlled\\
Unital Open Quantum Systems}

\author{Corey~O'Meara, Gunther~Dirr, and Thomas Schulte-Herbr{\"u}ggen$^*$%
\thanks{C. O'Meara and T. Schulte-Herbr{\"u}ggen are with the Department of Chemistry,
Technical University of Munich, 85747 Garching, Germany.
G. Dirr is with the Institute of Mathematics, University of W{\"u}rzburg, 97074 W{\"u}rzburg, Germany.
E-mail: tosh@ch.tum.de ($^*$to whom correspondence should be addressed)}
\thanks{This work was supported in part by the integrated EU programme Q"~ESSENCE,
the exchange with COQUIT,
and by {\em Deutsche Forschungsgemeinschaft} (DFG) in the
collaborative research centre SFB 631.
We also gratefully acknowledge
support and collaboration enabled within the two International
Doctorate Programmes of Excellence
{\em Quantum Computing, Control, and Communication} (QCCC) as well as
{\em Identification, Optimisation and Control with Applications in
Modern Technologies} by the Bavarian excellence network ENB.}}


\IEEEspecialpapernotice{ {\small dated: Aug.~15. 2011} }

\maketitle


%
\IEEEpeerreviewmaketitle

\begin{abstract}
\boldmath
We extend standard Markovian open quantum systems (quantum channels)
by allowing for
Hamiltonian controls and elucidate their geometry in terms of
Lie semigroups. For standard dissipative interactions with
the environment and different coherent controls, we particularly
specify the tangent cones (Lie wedges) of the respective Lie
semigroups of quantum channels.
These cones are the counterpart of the infinitesimal generator of a
single one-parameter semigroup. They comprise all directions the
underlying open quantum system can be steered to and thus give
insight into the geometry of controlled open quantum dynamics.
Such a differential characterisation is highly valuable for
approximating reachable sets of given initial quantum states
in a plethora of experimental implementations.
\end{abstract}

\tableofcontents

\section{Introduction}
Extending quantum channels by allowing for Hamiltonian control
turns them into interesting and important examples of geometric control of
open systems. While for \emph{closed systems}, the theory of Lie
groups provides a rich structure to address questions of reachability,
accessibility, and controllability \cite{Jurdjevic97}, already
simple \emph{open quantum systems} come with the intricate
geometry of Lie semigroups \cite{HHL89,Lawson99}.
For instance, in most closed systems
the reachable set to an initial state $\rho_0$ simply is 
the orbit $\mathcal O_{\bG}(\rho_0):=\{G\rho_0 G^{-1}\,|\, G\in\bG\}$
of a unitary subgroup $\bG$ whose Lie algebra can be identified
easily via Lie closure, while in {\em open systems} reachable sets are much
more difficult to determine explicitly.
Thus in view of controlling open quantum dynamics, in \cite{DHKS08}
we systematically related the framework of completely positive semigroups
\cite{Kraus71,Koss72,Koss72b,Choi75,GKS76,Lind76,Kraus83}, which is
well established in quantum physics, with the more recent
mathematical theory of Lie semigroups. An early example confined to
single-qubit systems can be found in \cite{Alt03}.

More precisely, for exploiting the power of systems and control theory
in open quantum dynamics, the system parameters have to be characterised
first, e.g., by input-output relations in the sense of quantum
process tomography. The decision problem whether the dynamics
of the quantum system thus specified is Markovian to good
approximation has recently been analysed \cite{Wolf08a,Wolf08b}.
Moreover (time-dependent) Markovian quantum channels
were elucidated from the viewpoint of divisibility \cite{Wolf08a} 
thus paving the way to Lie semigroups \cite{DHKS08}. Following up,
this work sets out to determine the geometry of quantum channel
semigroups in terms of their tangent cones (Lie wedges)
for a number of coherently controlled standard unital channels
in a unified frame in line with \cite{KuDiHeTAC11a}.

For the first time, here we 
explicitly parameterize the {\em set of all possible directions}
an open quantum system under coherent controls
may take --- its \emph{Lie wedge}. Thereby, we heavily exploit
the fact that the set of all reachable quantum maps governed
by a controlled Markovian master equation constitutes a Lie semigroup
\cite{DHKS08}.
Previous characterizations of reachable sets
for unital open quantum systems by majorization techniques,
e.g., \cite{Yuan10}, become increasingly inaccurate
once full controllability of the Hamiltonian part
(\mbox{condition (H)} {\em vide infra})
is violated, which for growing number of qubits happens
{\em in all experimentally realistic settings}. 
In contrast, the Lie-semigroup tools presented here
do not require \mbox{condition (H)}
and carry over to multi-qubit systems 
without the draw-back of increasing inaccuracy. 

\section{Theory and Background}
We start out by recalling some basic notions and notations
of Lie subsemigroups \cite{HHL89} and their
application for characterising reachable sets of quantum control
systems modelled by Lindblad-Kossakowski master equations
\cite{DHKS08}.

\subsection{Lie Semigroups}\label{sec:LieSemiBasics}
To begin with, let $\bG$ be a matrix Lie group, i.e.~a
group which is (isomorphic to) a path-connected subgroup of
$GL(n,\R{})$ or $GL(n,\C{})$ for some $n \in \N{}$, and let
$\fg$ be its corresponding matrix Lie algebra. Thus $\fg$ is
(isomorphic to) a Lie subalgebra of $\gl(n,\R{})$ or $\gl(n,\C{})$.
Then a subset $\bS\subset\bG$ which is closed under the
group operation in the sense $\bS\cdot\bS\subseteq\bS$ and which
contains the identity $\unity$ is said to be a {\em subsemigroup}
of $\textbf{G}$.
The largest subgroup within $\bS$ is written $E(\bS):=\bS\cap\bS^{-1}$.

Furthermore, a closed convex cone  $\fw \subset \fg$
is called a wedge. The largest linear subspace of $\fw$
is denoted $E(\mathfrak{w}):=\mathfrak{w}\cap(\mathfrak{-w})$
and it is termed the {\em edge of the wedge} $\fw$.
Now, $\fw\subseteq\fg$
is a {\em Lie wedge} of $\fg$ if it is invariant under
the adjoint action 
of the subgroup generated by the edge $E(\mathfrak{w})$,
i.e.~if it satisfies
\begin{equation}
e^A \,\fw\, e^{-A} = \fw
\end{equation}
(or equivalently $e^{\adr_A}(\fw) = \fw$) for all $A\in E(\fw)$.
Note that the edge of a Lie wedge always
forms a Lie subalgebra of $\fg$.

Moreover, for any closed subsemigroup $\bS$ of $\bG$ we define
its tangent cone $L(\bS)$ at the identity $\unity$ by
\begin{equation}
L(\bS) := \{ A\in\fg\,|\, \exp(tA)\in \bS\text{\; for all\;} t\geq 0\}\;.
\end{equation}
Then one can show that $L(\bS)$ is a Lie wedge of
$\fg$ satisfying the identity
$
E\big(L(\bS)\big) = L\big(E(\bS)\big) .
$
Yet, the \/`local-to-global\/' correspondence
between Lie wedges and closed connected subsemigroups is
much more subtle than the correspondence between Lie (sub)algebras
and Lie (sub)groups: for instance, several connected subsemigroups may
share the same Lie wedge $\fw$ in the sense that $L(\bS) = L(\bS')$
for $\bS \neq \bS'$, or conversely there may be Lie wedges $\fw$
which do not correspond to any subsemigroup, i.e.~$\fw=L(\bS)$
fails for all subsemigroups $\bS\subset\bG$.

Therefore, one introduces the important notion of
a {\em Lie subsemigroup} $\bS$ which is characterised by the equality
\begin{equation}
\bS = \overline{\expt{\exp L(\bS)}}_S\;,
\end{equation}
where the closure is taken in $\bG$ and $\expt{\exp L(\bS)}_S$
denotes the subsemigroup 
generated by $\exp L(\bS)$,
i.e.~$\expt{\exp L(\bS)}_S := \{e^{A_1}
\cdots e^{A_n}\,|\, n\in\N, \, A_1, \dots, A_n\in L(\bS)\}$.
Moreover, a Lie wedge $\fw$ is said to be {\em global} in $\bG$,
if there is a Lie subsemigroup $\bS\subset\bG$ such that
\begin{equation}
L(\bS)=\fw
\;.
\end{equation}
Thus, one has the identity $\bS=\overline{\expt{\exp \fw}}_S$.

Whenever a Lie wedge $\fw\subset\fg$ specialises to be
compatible with the Baker-Campbell-Hausdorff (BCH)
multiplication
\begin{equation}
A\star B := A+B+\tfrac{1}{2}[A,B]+ \dots
= \log(e^A e^B)\quad\forall A,B\in\fw\end{equation}
defined via the BCH series, it
is termed {\em Lie semialgebra}. For this to be the case, there
has to be an open BCH neighbourhood $\mathcal B\subset\fg$
of the origin in $\fg$ such that
$(\fw\cap \mathcal B)\star(\fw\cap \mathcal B) \subseteq\fw$.
An equivalent definition for 
being a Lie semialgebra is given by the tangential condition
\begin{equation}\label{eqn:semialg-incl}
[A,T_A\fw] \subset T_A\fw \quad\text{for all $A \in \fw$}\;,
\end{equation}
where $T_A\fw$ denotes the tangent space of $\fw$ at $A$
defined by
\begin{equation}
T_A\fw := \big(A^\perp \cap \fw^*\big)^\perp\;.
\end{equation}
Here $A^\perp$ denotes the {\em orthogonal complement} of $A$ and
$\fw^* :=
\{A \in \fg \,|\, \langle A,B \rangle \geq 0 \text{ for all } B \in \fw\}$
the {\em dual wedge}---both taken with respect to the standard trace
inner product.
The conceptual importance of Lie semialgebras roots in
the fact that---in Lie semialgebras---the exponential map of
a zero-neighbourhood in $L(\bS)$ yields a $\unity$"~neighbourhood
in $\bS$.
In contrast, as soon as $\fw$ is merely a Lie wedge that fails
to carry the stronger structure of a Lie semialgebra, there will be
elements in $\bS$ that are arbitrary close to the identity
without belonging to any one-parameter semigroup completely
contained in  $\bS$. For more details and a variety of
illustrative examples, we recommend \cite{HHL89} and
\cite{HofRupp97div}, where the respective introduction
does provide a lucid overview of the entire subject.
The connection between Lie semialgebras and time-independent Markovian
quantum channels has been worked out in detail in \cite{DHKS08}.

With these stipulations, the frame is set to describe the time
evolution of Markovian (i.e., memory-less) open quantum
systems in the differential geometric picture of Lie wedges.


\subsection{Markovian Quantum Dynamics and Quantum Channels}

Markovian quantum dynamics is conveniently described by
a linear autonomous differential equation
\begin{equation}\label{eqn:LE}
\dot X(t) = - \mathcal L \, X(t)\;, 
\end{equation}
where $X(t)$ usually denotes the state of a quantum system
represented by its density operator $\rho(t)$,
i.e.~$\rho(t)=\rho(t)^\dagger$, $\rho(t) \geq 0$, and
$\tr\rho(t)=1$. Here and henceforth, $(\cdot)^\dagger$
denotes the adjoint (complex-conjugate transpose).
For ensuring complete positivity,
$\mathcal L$ has to be of Lindblad form
\cite{Lind76}, i.e.
\begin{equation}\label{eqn:master-rho-uncontrolled}
\mathcal L(\rho) = i\adr_{H}(\rho)+\Gamma_L(\rho)\;,
\end{equation}
with $\adr_{H_j}(\rho):= [H_j, \rho]$ and
\begin{equation}\label{eqn:GKS}
\Gamma_L(\rho) := \tfrac{1}{2}\sum_kV_k^\dagger V_k\rho
+ \rho V_k^\dagger V_k-2V_k\rho V_k^\dagger\;,
\end{equation}
Here, the {\em Hamiltonian} $H$ is assumed to be a Hermitian
\mbox{$N \times N$} matrix while the {\em Lindblad generators}
$\{V_k\}$ may be arbitrary \mbox{$N\times N$} matrices.
The resulting equation of motion \eqref{eqn:LE}
acts on the vector space of all Hermitian operators,
$\herm(N)$, and more precisely, leaves the set of all
density operators 
$\pos_1(N) := \left\{\rho\in\herm(N)\,|\,\rho=\rho^\dagger,\rho\geq0, \tr\rho=1\right\}$ 
invariant.

\medskip

In \cite{DHKS08} it was shown that the set of all Lindblad
generators $\{-\mathcal{L}\}$ has an interpretation as a particular
Lie wedge. To see this, consider the group lift of \eqref{eqn:LE},
i.e.~now $X(t)$ denotes an element in the general linear group
$GL(\herm(N))$. Moreover, define the set of all completely
positive (cp), trace-preserving invertible linear operators
acting on $\herm(N)$ as $\bP^{cp}$, i.e.
\begin{equation*}
\bP^{cp} := \{T\in GL(\herm(N))\,|\, T \;\text{is cp and trace-preserving}\}
\end{equation*}
and let $\bP^{cp}_0$ denote its connected component of the identity.
Then, $\bP^{cp}$ is exactly the set of so-called invertible 
\emph{quantum channels}. A quantum channel $T$ is said to
be \emph{time independent Markovian} or briefly \emph{Markovian},
if it is a solution of \eqref{eqn:LE}. Thus $T = {\rm e}^{-t \mathcal{L}}$
for some fixed  Lindblad generators $\mathcal{L}$ and
some $t \geq 0$. Furthermore, $T$ is \emph{time dependent Markovian}
if it is a solution of \eqref{eqn:LE}, where now
$\mathcal{L} =  \mathcal{L}(t)$ may vary in time
(for terminology see also \cite{Wolf08a,Wolf08b}).
Finally, we will denote 
the set of all {\em time independent Markovian} and
{\em time dependent Markovian} quantum channels
by $\mathbf{MQC}$ and $\mathbf{TMQC}$, respectively. 
Then, with regard to the work by Lindblad \cite{Lind76} and
Kossakowski \cite{GKS76}, one obtains the following result
\cite{DHKS08}:
\begin{enumerate}
\item[(a)]
The global {\em Lie wedge} of $\bP^{cp}_0$ is given by the
{\em set of all Lindblad generators} of the form
\begin{equation}
- \mathcal L:= -\big(i\adr_{H} + \Gamma_L\big)
\end{equation}
with $H \in \herm(N)$ and  $\Gamma_L$ as in \eqref{eqn:GKS}.
\item[(b)]
The Lie semigroup
\begin{equation}
\overline{\expt{\exp L(\bP^{cp}_0)}}_S
\end{equation}
clearly contains $\mathbf{MQC}$ and moreover it exactly coincides with the closure of
$\mathbf{TMQC}$
thus excluding the {\em non}\/-Markovian ones in $\bP^{cp}_0$,
which is most remarkable.
\end{enumerate}

\noindent
While assertion (a) reformulates previous results by Lindblad
and Kossakowski \cite{Koss72,GKS76,Lind76}, part (b) is noteworthy
as it also says that $\bP^{cp}_0$ is {\em not} a Lie subsemigroup of
$GL(\herm(N))$.

\subsection{Coherently controlled Master Equations}
\label{subsec:notions}

\emph{Controlled} Markovian quantum dynamics is appropriately
addressed as right-invariant {\em bilinear control system}
\cite{DHKS08,dAll08,DiHeGAMM08,Elliott09}
\begin{equation}\label{eqn:sigma}
\dot \rho(t) = -\mathcal L_{u(t)}\big(\rho(t)\big)\quad,
\quad \rho(0) \in \pos_1(N)\;,
\end{equation}
where $\mathcal{L}_u$ now depends on some control variable
$u \in \R{}^m$.

Here, we focus on \emph{coherently controlled}
open systems. This means that $\mathcal{L}_u$ has the
following special from 
\begin{equation}\label{eqn:master-rho}
\mathcal{L}_u(\rho)  = -i\adr_{H_u}(\rho)-\Gamma_L(\rho)\quad\text{with}
\end{equation}
\begin{equation}\label{eqn:controlHam}
\adr_{H_u} := \adr_{H_d}+\sum_{j=1}^mu_j\adr_{H_j}\;.
\end{equation}
Note that the control terms $i\adr_{H_j}$ with {\em control Hamiltonians}
$H_j \in \herm(N)$ are usually switched  by piecewise constant
{\em control amplitudes} $u_j(t)\in\R{}$.
The drift term of \eqref{eqn:master-rho} is then composed of two parts, (i)
the term $i\adr_{H_d}$ (in abuse of language sometimes called
\/`Hamiltonian\/' drift) accounting for the coherent time evolution
and (ii) a dissipative Lindblad part $\Gamma_L$. 
So $\mathcal{L}_u$ denotes the {\em coherently controlled Lindbladian}.
As in the uncontrolled case,
system \eqref{eqn:sigma} acts on
the vector space of all Hermitian operators leaving the set of all
density operators invariant. Equivalently, one can regard
\eqref{eqn:sigma} as an affine system on 
$\mathfrak{her}_0(N) := \{ H \in \herm(N)\;|\; \tr H = 0\}$.

In the following, we further impose {\em unitality}, i.e.
we assume $\Gamma_L(\unity)=0$. This ensures that
\eqref{eqn:sigma} actually yields  a bilinear control system on
$\mathfrak{her}_0(N)$ instead of an affine one. 
Therefore, it allows a group lift to
$GL\big(\mathfrak{her}_0(N)\big)$ which henceforth is
referred to as ($\Sigma$), i.e.
\begin{equation}\label{eqn:SIGMA}
(\Sigma)\quad
\dot X(t) = -\mathcal L_{u(t)} \, X(t) \,,\;
X(0) \in GL\big(\mathfrak{her}_0(N)\big)\;.
\end{equation}
The corresponding group lift in the affine case is more involved
\cite{DiHeGAMM08,DHKS08}. 
Now, the {\em system semigroup} $\bP_\Sigma$ associated to ($\Sigma$)
reads
\begin{eqnarray}\label{eqn:Psemi}
\bP_\Sigma=\langle T_u(t)=
\exp(-t \mathcal L_u)\,|\,t\geq 0, u \in \R{}^m \rangle_S
\end{eqnarray}
and lends itself to exemplify the notion of a Lie wedge.
%
%
%
%
To distinguish between different notions of controllability in open 
systems, we define three algebras: the \emph{control algebra} $\fk_c $,
the \emph{extended algebra} $\fk_d$, and the \emph{system algebra}
$\fs$ as follows
\begin{equation}\label{eqn:HamL-closure}
\begin{split}
\fk_c 
& := \expt{i\adr_{H_j}\,|\, j=1,\dots, m}_{\sf Lie}\;,\\[2mm]
\fk_d
& := \expt{i\adr_{H_d}, i\adr_{H_j}\,|\, j=1,\dots, m}_{\sf Lie}\;,\\[2mm]
\fs
&:= \expt{\mathcal L_u|\, u_j\in\mathbb R }_{\sf Lie}\\[2mm]
&= \expt{i\adr_{H_d}+\Gamma_L, i\adr_{H_j}|\, j=1,\dots, m}_{\sf Lie}\;.
\end{split}
\end{equation}
Note that $\fs$ is different from $\fk_d$, because it
contains the entire
drift term $(i\adr_{H_d}+\Gamma_L)$ for the Lie closure,
while $\fk_d$ only takes its Hamiltonian component $i\adr_{H_d}$.
Then $(\Sigma)$ is said to fulfill condition (H), (WH), and (A),
respectively, if
\begin{eqnarray}
&(H) \qquad
\fk_c &=\;  \adr_{\su(N)} \\
&(W\negthinspace H)\quad
\fk_d &=\;  \adr_{\su(N)} \;\text{while}\;\;\fk_c \neq \adr_{\su(N)}\\
&(A)\qquad
\fs &=\; \gl(\herm_0(N))
\end{eqnarray}
While condition (A) respects a standard construction of
non-linear control theory \cite{Jurdjevic97,JS72} to express
accessibility, conditions (H) and (WH)
serve to  characterize different types of {\em controllability}
of the Hamiltonian part of ($\Sigma$) in the absence
of relaxation: Condition (H) says that the Hamiltonian part is
fully controllable even {\em without} resorting to the drift
Hamiltonian, whereas condition (WH) yields full controllabilty
of the Hamiltonian part with the drift Hamiltonian being {\em necessary}.
We refer to the first scenario as \emph{(fully) $H$-controllable}
and to the second as satisfying the (WH)-condition.
Generically, open systems ($\Sigma$) given by \eqref{eqn:SIGMA}
meet the accessibility condition (A) \cite{Alt04,Diss-Indra}.

%
Finally, note that via
$\e^{i\adr_{H}}(\rho) = \e^{i H} \rho\, \e^{-iH}$
the Lie algebra $\adr_{\su(N)}$ generates the Lie group
$\Adr_{SU(N)} \iso PSU(N)$ here acting on $\herm_0(N)$
by conjugation.


\subsection[Lie Wedges for Coherently Controlled GKS"~Master
Equations]{Computing Lie Wedges for Controlled
Master Equations}
\label{subsec:LW-comp}

Here, the goal is to determine the (global) Lie wedge
of a coherently controlled unital open system ($\Sigma$) given in terms
of its Markovian master equation \eqref{eqn:SIGMA} of GKS"~Lindblad
form. 
In view of the examples worked out in detail in Sec.~\ref{sec:geoR3},
here we sketch how to approximate a Lie wedge of a controlled
Markovian systems in two ways, (i) by an {\em inner approximation}
and (ii) by an {\em outer approximation} thus following
\cite{Lawson99,DHKS08}.
Moreover for unital systems, we present two results
which guarantee that the {\em inner approximation} is
global and thus coincides with the Lie wedge
$L(\overline{\bP}_\Sigma)$ sought for.

Let ($\Sigma$) be a unital open control system as in
\eqref{eqn:SIGMA} where, for simplicity, the system algebra
$\fs$ fulfills the accessibility condition (A).
Moreover, let
\begin{equation}
\Omega_\Sigma := \{\cLu | u \in \R{}^m\} \subset
\gl(\mathfrak{her}_0(N))
\end{equation}
be the set of all directions specified by \eqref{eqn:SIGMA}.
The {\em reachable set} $\Reach(\Omega_\Sigma,\unity)$
of ($\Sigma$) is defined
as the set of all states $X(T)$, $T\geq 0$ that can be reached
from the unity $X(0)=\unity$ under the dynamics of ($\Sigma$),
while the controls $u(t) \in \R{}^m$ are assumed to be
piecewise constant functions. In general, one could
allow for larger classes of {\em admissible} controls,
such as locally bounded or locally integrable ones. Yet, the closure
of the corresponding reachable sets will not differ
\cite{JS72,S83,Elliott09}.

Clearly, $\Reach(\Omega_\Sigma,\unity)$ takes the form of a
subsemigroup within the embedding Lie group $GL(\mathfrak{her}_0(N))$
in the sense of Sec.~\ref{sec:LieSemiBasics}.  For instance,
restricting the control amplitudes $\{u_j\}$ to be piecewise constant yields the equality
$\Reach(\Omega_\Sigma,\unity) = \bP_\Sigma$. 
More generally,
the following result holds.

\medskip
\begin{theorem}[\cite{Lawson99}]\label{thm:L-saturate}
Let $\bP_\Sigma$ be defined as in \eqref{eqn:Psemi}.
Then
\begin{equation}\label{reach}
\overline{\bP}_\Sigma \;=\; \overline{\Reach(\Omega_\Sigma,\unity)}
\;=\; \overline{\Reach(L(\bP_\Sigma),\unity)}\;.
\end{equation}
In particular, $\overline{\bP}_\Sigma$ is a Lie subsemigroup.
Furthermore, $L(\overline{\bP}_\Sigma)$ is the smallest global
Lie wedge containing $\Omega_\Sigma$ as well as
the largest subset $\Omega'$ of $\gl(\mathfrak{her}_0(N))$
which satisfies the equality
\begin{equation}
\overline{\Reach(\Omega',\unity)} \;=\;
\overline{\Reach(\Omega_\Sigma,\unity)}\;.
\end{equation}
%
%
Due to the last property, 
the Lie wedge $L(\overline{\bP}_\Sigma)$
is also called the {\em Lie saturate} of $\Omega_\Sigma$,
cf.~\cite{JurdKupka81a,JurdKupka81b,Lawson99}.
\end{theorem}


\medskip
Unfortunately, for an arbitrary system ($\Sigma$), currently no
procedure is known to explicitly determine its global Lie wedge.
Yet there is a straightforward strategy to compute
an {\em inner approximation} \cite{Lawson99,DHKS08}.
It consists of the following steps:

\begin{enumerate}
\item[(1)]
form the smallest closed convex cone $\fw$ containing $\Omega_\Sigma$;
\item[(2)]
compute the edge $E(\fw)$ of the wedge and the
smallest Lie algebra $\fe$ containing $E(\fw)$,
i.e.~$\fe :=\langle E(\fw) \rangle_{\rm Lie}$;
\item[(3)] make the wedge invariant under the $\Adr$"~action
of $\mathfrak{e}$ by forming the set
$\bigcup_{A \in \mathfrak{e}} \Adr_{\exp A}(\fw)$;
\item[(4)]
update by taking the convex hull $\conv\{S\}$ of the set $S$ obtained in step (3);
\item[(5)]
repeat steps (2) through (4) until nothing new is added:
the resulting final wedge $\fw_0$ is henceforth referred
to as {\em inner approximation} to the global Lie wedge
$L(\overline{\bP}_\Sigma)$.
\end{enumerate}

Now, the crucial question arises whether the inner approximation
$\fw_0$ is global or not. If it is global, Theorem~\ref{thm:L-saturate}
guarantees that $\fw_0$ is equal to $L(\overline{\bP}_\Sigma)$.
Next we present two results which proved quite helpful to decide
the globality problem: The first one 
yields a global {\em outer approximation} $\fw^0$ of
$L(\overline{\bP}_\Sigma)$.
Combining inner and outer approximation, the Lie wedge
$L(\overline{\bP}_\Sigma)$ sought for can be determined via the
inclusions
\begin{equation}
\fw_0 \subseteq L(\overline{\bP}_\Sigma) \subseteq \fw^0.
\end{equation}
Clearly, if the outer and inner approximations coincide, one is done.
The second one 
based on the so-called \emph{Principal Theorem of Globality}
from \cite{HHL89} (see also Appendix~A) 
provides
a \/`direct\/' method for proving globality. It will be the key
tool to show that the inner approximations given in the worked
examples of Secs.~\ref{sec:geoR3} and \ref{sec:q-channels} are
in fact {\em global} Lie wedges.

\medskip
\begin{theorem}[\cite{DHKS08}]\label{thm:LW-outer}
Let ($\Sigma$) be a unital controlled open system
as in \eqref{eqn:SIGMA}. If there exists a \emph{pointed}
cone $\mathfrak{c}$ in the set of all positive semidefinite
operators that act on $\mathfrak{her}_0(N)$ so that
\begin{enumerate}
\item[(1)] $\Gamma_L\in\fc$\\[-3mm]
\item[(2)] $[\fc,\fc]\subset\adr_{\su(N)}$\\[-3mm]
\item[(3)] $[\fc,\adr_{\mathfrak{su}(N)}]\subset (\fc-\fc)$\\[-3mm]
\item[(4)] $\Adr_U\fc\Adr_{U^\dagger}\subset \fc$ for all $U\in SU(N)$,
\end{enumerate}
then the subsemigroup associated to ($\Sigma$) follows the inclusion
$\overline{\bP}_\Sigma\subseteq\Adr_{SU(N)}\cdot\exp(-\mathfrak{c})$
and hence its Lie wedge obeys the relation 
$L(\overline{\bP}_\Sigma)\subseteq\adr_{\mathfrak{su}(N)}\oplus(-\mathfrak{c})$,
i.e.~$\adr_{\su(N)}\oplus(-\fc)$ is a
global outer approximation to $L(\overline{\bP}_\Sigma)$.
\end{theorem}


\medskip
\begin{corollary}(\cite{DHKS08,Alt03})\label{HcontrolGlobalWedge}
Let ($\Sigma$) be a unital 
single-qubit system satisfying condition (H) 
with a generic\footnote{In \cite{DHKS08} Cor.~\ref{HcontrolGlobalWedge} 
is stated under the above genericity
assumption; yet one can drop this additional condition.} 
Lindblad term $\Gamma_L$.
Then $\overline \bP_{\Sigma}=\Adr_{SU(2)}\cdot\exp(-\fc)$,
where the cone 
\begin{eqnarray*}
\fc:=\R{}_0^+\conv\left\{\Adr_U\,\Gamma_L\,\Adr_{U^\dagger} \;|\;
U\in SU(2)\right\}
\end{eqnarray*}
is contained in the set of all positive semidefinite elements in
$\mathfrak{gl}(\mathfrak{her}_0)$. Furthermore
$L(\overline{\textbf{P}}_{\Sigma})=
\adr_{\mathfrak{su}(2)}\oplus(-\mathfrak{c})$.
\end{corollary}

\medskip
\begin{theorem}\label{thm:globality-3}
Let ($\Sigma$) be a unital controlled open system
given by \eqref{eqn:SIGMA}. In addition assume that
($\Sigma$) meets the accessibility condition (A) 
and that the Lie subgroup $\bK$ which corresponds
to the control algebra $\fk_c$ is closed within $SU(N)$. Then,
$\fw := \fk_c \oplus (-\fc)$ is a {\em global} Lie wedge in
$\gl(\mathfrak{her}_0(N))$, where
$
\fc :=
\R{}_0^+\conv \big\{
\Adr_{U}\,\big(\adr_{{\rm i}H_d}+\Gamma_L\big)\,\Adr^\dagger_{U}
\;|\; U \in \bK \big\}\;.
$
Moreover, $\fw$ is the global Lie wedge of ($\Sigma$), i.e.
$\fw = L(\overline{\bP}_\Sigma)$.
\end{theorem}

\begin{proof}(Sketch)
The full proof will be given elsewhere in a more general context.
For applying the \/`Principal Theorem of Globality\/' \cite{HHL89}
(see Appendix~A),
the following steps have to be established:
\begin{enumerate}
\item[(1)]
The edge of $\fw$ coincides with $\fk_c$.
\item[(2)]
$\fw$ is a Lie wedge in $\fg := \gl(\mathfrak{her}_0(N))$.
\item[(3)]
There exists a function $\varphi:GL(\mathfrak{her}_0(N)) \to \R{}$
such that its differential satisfies
${\rm d}\varphi(X) \, AX \geq 0$ for all
$X \in GL(\mathfrak{her}_0(N))$ and all $A \in \fw$.
\item[(4)]
The differential of $\varphi$ fulfills
${\rm d}\varphi(\unity) \, A > 0$ for all
$A \in \fw\setminus E(\fw)$. 
\end{enumerate}
Note that step (3) is the essential one, and an appropriate candidate
for $\varphi$ is given by
$
X \mapsto - \langle X,X \rangle
:=
- \sum_{k=1}^{N^2-1}\tr\big(X(B_k)X(B_k)\big)
\;,
$
where $B_1, \dots, B_{N^2-1}$ is any orthonormal basis of
$\mathfrak{her}_0(N)$.
\end{proof}
\medskip

\noindent
As a useful tool, 
we add the following Corollary,
which is put into a broader context in Appendix~A:
\medskip
\begin{corollary}[\cite{HHL89}]\label{HHLglobalcor}
%
Let $\bG$ be a Lie group with Lie algebra $\fg$ and let
$\fw_0\subseteq\fw$ be two Lie wedges in $\fg$.
Provided one has
$
   \fw_0\setminus-\fw_0\subseteq\fw \setminus -\fw 
$ [or equivalently $E(\fw_0)=E(\fw)\cap\fw_0$],
then $\fw_0$ is global in $\bG$ if the following conditions are
satisfied:
(i)
$\fw$ is global in $\bG$;
(ii)
the edge of $\fw_0$ 
is the Lie algebra of a closed Lie subgroup of $\bG$.
\end{corollary}
\subsection*{Guideline through Applications}

For illustrating the power of the Lie-semigroup formalism by
applications, we follow a two-fold route: Sec.~\ref{sec:geoR3}
addresses three paradigmatic types of bilinear control systems on
$\mathbb R^3$, where the control parts of the dynamics generate
easy-to-visualise rotations in $SO(3)$.
{\em Thus Sec.~\ref{sec:geoR3} is meant to be readable without any background in
quantum mechanics}, yet it directly corresponds to single-qubit systems
undergoing relaxation as the presented examples coincide with the
so-called {\em coherence-vector representation} of such
systems \cite{AlickiLendi87}.
Therefore, the results obtained in Sec.~\ref{sec:geoR3} can readily be transferred
to Sec.~\ref{sec:q-channels}, where we address quantum channels
in the customary explicit $\su(2)$"~representation of qubits. By the isomorphism
$\so(3)\iso\su(2)$, the geometry in Sec.~\ref{sec:geoR3} thus illustrates
key results in Sec.~\ref{sec:q-channels} for qubit channels.

\section{Geometry of Open Systems in $\mathbb R^3$}\label{sec:geoR3}

In this section, we discuss three simple introductory examples of \/`open\/'
systems, the geometry of which can be envisaged as rotations in $\mathbb R^3$
concomitant to relaxation. To fix notations, define the following
\begin{equation}\label{eqn:R3rot_gen}
H_x:=\left[\begin{smallmatrix} 0 &0 &0 \\ 0 &0 &-1 \\ 0 &1 &0 \end{smallmatrix}\right]\;,\quad
H_y:=\left[\begin{smallmatrix} 0 &0 &1 \\ 0 &0 &0 \\ -1 &0 &0 \end{smallmatrix}\right]\;,\quad
H_z:=\left[\begin{smallmatrix} 0 &-1 &0 \\ 1 &0 &0 \\ 0 &0 &0 \end{smallmatrix}\right]\;\quad
\end{equation}
as generators of the rotations $R_\nu(\theta):={\rm e}^{\theta H_\nu}$ reading
\begin{equation}\label{eqn:R3rots}
\begin{split}
R_x(\theta) &:=\left[\begin{smallmatrix} 1 &0 &0 \\ 0& \ct &-\st \\ 0 &\st &\ct  \end{smallmatrix}\right],\;
R_y(\theta) :=\left[\begin{smallmatrix} \ct &0 &\st \\ 0 &1 &0 \\ -\st &0 &\ct \end{smallmatrix}\right],\\[2mm]
R_z(\theta) &:=\left[\begin{smallmatrix} \ct &-\st &0 \\ \st &\ct &0 \\ 0 &0 &1 \end{smallmatrix}\right]\;.
\end{split}
\end{equation}
So we have $\expt{H_x,H_y,H_z}_{\sf Lie}=\so(3)$ and thereby a basis for the skew-symmetric matrices
forming the $\fk$-part in the Cartan decomposition
$\gl(3,\mathbb R)=\so(3)\oplus\mathfrak{sym}(3)$, where the $\fp$-part is spanned
by the symmetric matrices
\begin{equation}\label{eqn:R3sym_gen}
p_x:=\left[\begin{smallmatrix} 0 &0 &0 \\ 0 &0 &1 \\ 0 &1 &0 \end{smallmatrix}\right]\;,\quad
p_y:=\left[\begin{smallmatrix} 0 &0 &1 \\ 0 &0 &0 \\ 1 &0 &0 \end{smallmatrix}\right]\;,\quad
p_z:=\left[\begin{smallmatrix} 0 &1 &0 \\ 1 &0 &0 \\ 0 &0 &0 \end{smallmatrix}\right]\;\quad
\end{equation}
and the diagonal $3 \times 3$-matrices $E_{ii} := e_ie_i^\top$ for $i=1,2,3$.
Recall that the skew-symmetric $\fk$-part and a symmetric $\fp$-part
obeying the usual commutator relations
$\comm \fk \fk \subseteq \fk$,
$\comm \fk \fp \subseteq \fp$ and $\comm \fp \fp \subseteq \fk$.
For later convenience, we note commutation relations for the above
basis in Tab.~\ref{tab:comm-tab}.

\subsection[{\em Example~1:} Fully H-Controllable System with General Relaxation Operator]
{{\bf Example~1}: Corresponds to a Qubit System with 
Condition (H) Satisfied and General Relaxation Operator}

Using definitions from above, consider the control system
in $GL(3,\mathbb R)$ given by the equation
\begin{eqnarray}\label{eq:ex1}
  \dot X=-(A+B_u)X \;,
\end{eqnarray}
where the control term $B_u:=u_xH_x+u_yH_y$ shall have independent
controls $u_x, u_y \in \R{}$, and the drift term
$A:=H_z+\Gamma_0$ is composed of a \/`Hamiltonian\/' component, $H_z$, and
a relaxation component given by the matrix
\begin{equation}\label{eqn:gamma-abc}
\Gamma_0:= \diag(a,b,c)
\end{equation}
with relaxation-rate constants $a,b,c\geq 0$.
Since $\expt{H_x,H_y}_{\sf Lie}=\so(3)$, system \eqref{eq:ex1}
satisfies in fact condition (H)
in the sense of Sec.~\ref{subsec:notions}
(i.e.~without resorting to the drift component $H_z$).


For explicitly computing the Lie wedge of \eqref{eq:ex1} we proceed
as in Sec.~\ref{subsec:LW-comp}. For the following calculations observe that
$H_x,H_y,H_z$ belong to the $\fk$-part, while $\Gamma_0$ is contained in
the $\fp$-part of \/`the\/' Cartan decomposition of $\gl(3,\R{})$.

\begin{table}[Ht!]
\caption{\label{tab:comm-tab}Commutation Table}
\begin{tabular}{c cccccc}
\hline\hline\\[-1.5mm]
$[H_\nu,(\cdot)]$ & $E_{11}$  & $E_{22}$ & $E_{33}$ & $p_x$ & $p_y$ & $p_z$\\[1mm]
\hline\\[-1mm]
$H_x$ &0 &$p_x$ &$-p_x$ &$-2\Delta_{23}$ &$-p_z$ &$p_y$ \\[0.5mm]
$H_y$ &$-p_y$ &$0$ &$p_y$ &$p_z$ &$-2\Delta_{31}$ &$-p_x$\\[0.5mm]
$H_z$ &$p_z$ &$-p_z$ &$0$ &$-p_y$ &$p_x$ &$-2\Delta_{12}$ \\[0.5mm]
\hline\hline\\[-1.5mm]
\multicolumn{3}{l} \mbox{define $\Delta_{ij}:=E_{ii}-E_{jj}$}
\end{tabular}
\end{table}

\medskip
\noindent
Step (1) of the algorithm gives the initial wedge approximation 
\begin{eqnarray}
\fw_1  := \R{}H_x \oplus \R{}H_y \oplus (-\fc_1)\;,
\end{eqnarray}
where $\fc_1 := \R{}^+_0 \big(H_z+\Gamma_0\big)$.
In step (2) one then readily finds
\begin{equation}
E(\fw_1) = \R{}H_x \oplus \R{}H_y
\end{equation}
so $\fe = \langle H_x, H_y\rangle_{\rm Lie} = \so(3)$.
Hence, step (3) and (4) give
\begin{equation}
\fw_0 = \so(3)\;\oplus\; \mathbb R^-_0\,
\conv\mathcal O_{SO(3)}(\Gamma_0) \;,
\end{equation}
where
$\mathcal O_{SO(3)}(\Gamma_0)
:=\{\Theta \Gamma_0 \Theta^\top\,|\, \Theta\in SO(3)\}$
denotes the {\em orthogonal orbit} of $\Gamma_0$. Here we used
the trivial fact that $\so(3)$ is $\Adr_{SO(3)}$-invariant.
By a well-known result of Uhlmann\footnote{The result mentioned
is originally stated for density matrices and $SU(N)$.
However, the proof in \cite{Ando89} immediately carries over
to symmetric matrices and $SO(N)$ \cite{Uhlm71,Ando89,MarshallOlkin}.}
the convex hull of the isospectral set $\mathcal O_{SO(3)}(\Gamma_0)$
simplifies to
\begin{equation}
\conv \mathcal O_{SO(3)} (\Gamma_0) =
\{S \in \mathfrak{sym}(3) \,|\, S \prec \Gamma_0\} =: \mathcal M(\Gamma_0)
\;.
\end{equation}
Defining the pointed convex cone
$\fc_0 := \mathbb R^+_0\,\mathcal M(\Gamma_0)$, we obtain
\begin{equation}
\fw_0 = \so(3)\;\oplus\; (-\fc_0)
\;,
\end{equation}
as final inner approximation to the global Lie wedge
of \eqref{eq:ex1}.

\medskip

\begin{lemma}\label{ex1edgeprop}
The set $\fw_0$ is a Lie wedge of $\mathfrak{gl}(3,\mathbb{R})$.
Its edge $E(\fw_0)$ is given by $\mathfrak{so}(3)$.
\end{lemma}

\medskip
\noindent
\begin{proof} It suffices to show
that the edge of $\fw_0$ is given by $\mathfrak{so}(3)$.
Then the invariance of $\fw_0$ under the $\Adr$-action
of $E(\fw_0)$ is obviously guaranteed by construction.
Clearly, one has the inclusion $\mathfrak{so}(3) \subset E(\fw_0)$.
Conversely, let $W \in E(\fw_0)$.
Then, $W=A+B$ with $A\in\mathfrak{so}(3)$ and $B\in\fc_0$.
Since $-W\in E(\fw_0)$, there exits $A'\in\mathfrak{so}(3)$
and $B'\in\fc_0$ such that $-W= A'+B'$. Hence
$A+B=-(A'+B')$ and thus $A+A'=-(B+B')$. Since $\fc_0\in\mathfrak{sym}(3)$, it
is trivial that $\mathfrak{so}(3)\cap\fc_0=\{0\}$. Therefore,
$A+A' = -(B+B')= 0$ and hence $A'=-A$ and $B'=-B$.
Now, $B,-B\in\fc_0=\mathbb R^+_0\,\mathcal M(\Gamma_0)$.
But $\mathcal{M}(\Gamma_0)$ is contained in the set of all positive
semidefinite matrices and therefore we conclude $B=0$. Thus we
obtain $E(\fw_0)=\mathfrak{so}(3)$.
\end{proof}


\medskip
\begin{proposition}\label{example1globality}
The set $\fw_0 = \so(3)\oplus (-\fc_0)$ is
the {\em global} Lie wedge to control system \eqref{eq:ex1}.
\end{proposition}

\medskip
\begin{proof}
This is an immediate consequence of Theorem~\ref{thm:globality-3}
and the fact that \eqref{eq:ex1} comes from a unital GKS"~Lindblad
master equation in the coherence-vector representation.
\end{proof}

\medskip
\begin{remark}
Alternatively to the above proof, one could apply Corollary
\ref{HcontrolGlobalWedge}, because $\fw_0$ is a matrix representation
of the global Lie wedge described therein. 
\end{remark}

\medskip
For general $\Gamma_0=\diag(a,b,c)$, the Lie wedge $\fw_0$ in
Example~1 does {\em not} carry the special structure of a Lie
{\em semialgebra}, cf.~Sec.~\ref{sec:LieSemiBasics}. This can be shown
by choosing a suitable $A' \in \fw_0$ which violates the 
inclusion $[A', T_{A'}\fw_0] \subset T_{A'}\fw_0$.
--- In contrast, for the case\footnote{Note that this case
exactly corresponds to the (isotropic) depolarising quantum channel
discussed in Sec.~\ref{sec:q-channels}.}
$\Gamma_0=\lambda\cdot\unity$,
indeed we obtain a
Lie semialgebra, because the BCH-product $A \star B$ obviously stays
inside $\fw_0$, whenever $A,B \in \fw_0$. Further details and
proofs are given in Appendix~B.

\subsection[{\em Example~2:} System Satisfying (WH)-Condition with Invariant Relaxation Operator]
{{\bf Example~2}: Corresponds to a Qubit System with Condition (WH) Satisfied 
and Control Invariant Relaxation Operator}

Consider the control system in $GL(3,\R{})$ given by
\begin{eqnarray}
\label{eq:ex2}
  \dot X=-(A+B_u)\,X\;,
\end{eqnarray}
where the control term $B_u$ is of the form $B_u:=u H_y$
with $u\in \R{}$ and the drift term
$A:=\Gamma_0+H_z$ is composed of a \/`Hamiltonian\/' part, $H_z$, and
a relaxation part
\begin{eqnarray}\label{eq:drift2}
\Gamma_0:=\gamma\,\diag(1,0,1)
\end{eqnarray}
with $\gamma\geq0$. Since 
$\langle H_y,H_z\rangle_{\sf Lie}=\mathfrak{so}(3)$
the system \eqref{eq:ex2} fulfills
condition (WH) in the sense of Sec.~\ref{subsec:notions}
but obviously not condition (H).

\medskip
Now in step (1) of the inner approximation procedure to the global
Lie wedge of \eqref{eq:ex2} one finds
\begin{equation}
\fw_1 = \mathbb R H_y \oplus  (-\fc_1)
\end{equation}
with $\fc_1 := \mathbb R_0^+ (H_z + \Gamma_0)$ whose edge is given by
\begin{equation}
  E(\fw_1)=\mathbb{R}H_y \;. 
\end{equation}

In step (3) we include elements obtained by
conjugations generated by edge elements identified
in step (2), i.e.\ elements of the form
\begin{equation}\label{dissipativeExample2}
{\rm e}^{\theta H_y} \,\fc_1\, {\rm e}^{-\theta H_y}
=\lambda\;(\cos(\theta)H_z+\sin(\theta)H_x+\Gamma_0)\qquad
\end{equation}
for $\theta\in\R{}$ and $\lambda\geq 0$.
By orthogonality \mbox{$\langle\Gamma_0|H_\nu\rangle=0$}
for $\nu=x,y,z$ one readily gets a Hilbert space
$\mathcal{H}:=\mbox{span}\left\{H_x,H_z,\Gamma_0\right\}$,
in which the edge-invariant cone elements of \eqref{dissipativeExample2}
can be expanded using the following short-hand 
\begin{eqnarray}\label{eqn:short-hand-scalprod}
\lambda\big(\sin(\theta)H_x+\cos(\theta)H_z+\Gamma_0\big) =:
\lambda\left[\begin{smallmatrix}\sin(\theta)\\ \cos(\theta)\\1\end{smallmatrix}\right]
\cdot\left[\begin{smallmatrix}H_x\\H_z\\\Gamma_0\end{smallmatrix}\right]\;.
\end{eqnarray}
Then its convex hull gives the final cone
\begin{eqnarray}\label{eqn:Ex2-cone}
  \fc_0:= \R{}^+_0\conv \left\{
\left[\begin{smallmatrix}\sin(\theta)\\ \cos(\theta)\\1\end{smallmatrix}\right]\cdot
\left[\begin{smallmatrix}H_x\\H_z\\\Gamma_0\end{smallmatrix}\right]
\;|\; \theta \in \R{}\right\}
\end{eqnarray}
--- a classical $3$-dimensional \/`ice cone\/', cf.~Fig.~\ref{fig:Ex2}(a).
By construction, $\fc_0$ remains $\Adr_{\exp E(\fw)}$-invariant.
Finally, since $H_y\perp\fc_0$,
the Lie wedge itself admits the orthogonal decomposition
\begin{eqnarray}\label{eqn:Ex2-cone-a}
  \fw_0:=\mathbb{R}H_y\oplus (-\fc_0)\;.
\end{eqnarray}

\begin{figure}
\begin{center}
\raisebox{35mm}{{\sf (a)}}
\includegraphics[scale=.9]{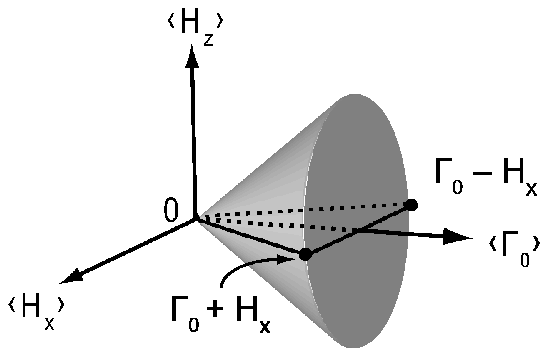}\\[2mm]
\raisebox{35mm}{{\sf (b)}}
\includegraphics[scale=.9]{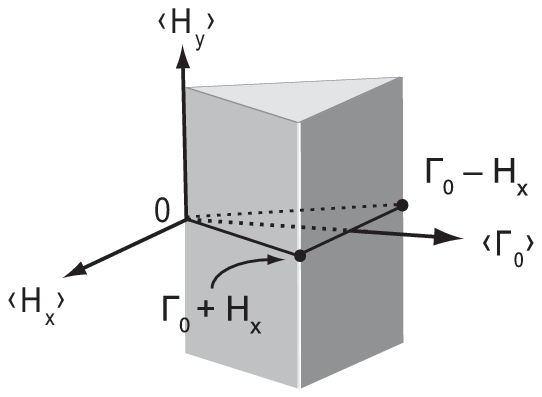}
\end{center}
\caption{\label{fig:Ex2} The Lie wedge for the system of {\bf Example~2}
satisfying the (WH)-condition is four-dimensional: it contains (a) the convex cone $\fc_0$ (see
\eqref{eqn:Ex2-cone} with $\lambda=1$)
generated by $e^{\theta H_y}(\Gamma_0+H_z)e^{-\theta H_y}$ and shown in the projection into
the subspace spanned by $\{H_x,H_z,\Gamma_0\}$; and it comprises (b) the
prism-shaped wedge projected into the subspace spanned by $\{H_y,H_z,\Gamma_0\}$.
Note that the {\em edge of the wedge} is spanned by the
control Hamiltonian $H_y$. Both parts of the figure scale with $\lambda\to\infty$.
}
\end{figure}

\medskip
\begin{proposition}\label{example2weakglobality}
The set $\fw_0 = \mathbb{R}H_y\oplus(-\fc_0)$ is
the {\em global} Lie wedge to control system \eqref{eq:ex2}.
\end{proposition}

\medskip
\begin{proof}
The Lie wedge property of $\fw_0$ can be derived as in Example 1.
Then the globality of $\fw_0$ follows again from Theorem
\ref{thm:globality-3}.
\end{proof}

\begin{remark}
For clarity, let us denote the Lie wedge of Example~1
for $\Gamma_0$ as in \eqref{eq:drift2} by $\fw'_0$
while $\fw_0$ still refers to the Lie wedge of Example~2.
Clearly, one has $\fw_0\subset\fw'_0$ and
$E(\fw_0) = E(\fw'_0)\cap \fw_0$.
Hence globality of $\fw_0$ also follows 
by Corollary \ref{HHLglobalcor} and
the globality of $\fw'_0$. 
\end{remark}

\medskip
Note that the Lie wedge $\fw_0$ in Example~2 does {\em not} specialise
to the form of a Lie {\em semialgebra} as can readily be verified
by a counter example:
According to \eqref{eqn:Ex2-cone},
choose $B\in\fw_0$ as $B=\Gamma_0+H_x$ 
(recalling $\Gamma_0:=E_{11}+E_{33}$ and $A=\Gamma_0+H_z$).
Then by the commutator relations of Tab.~\ref{tab:comm-tab},
the BCH product 
\begin{equation}
A\star B = 2\Gamma_0 + H_x + H_z + \tfrac{1}{2}(H_y + p_x + p_z)\, +\, \dots
\end{equation}
immediately leads outside the Lie wedge $\fw_0$, e.g.,
by the non-vanishing component $p_x + p_z$.
(NB: This argument can be made rigorous by introducing a scaling factor $t$
to give $tA\star tB$).

Finally, the edge of $\fw_0$ in Example~2
is $E(\fw_0)=\mathbb{R}H_y$ (for $\gamma>0$, see Fig.~\ref{fig:Ex2})
while in the limit of a closed system, i.e.\ for $\gamma=0$, it turns into 
the entire Lie algebra $\mathfrak{so}(3)$.

\subsection[{\em Example~3:} System Satisfying (WH)-Condition with General Diagonal Relaxation Operator]
{{\bf Example~3}: Corresponds to a Qubit System with Condition
(WH) Satisfied and General Diagonal Relaxation Operator}

Consider the contol system in $GL(3,\R{})$ given by
\begin{eqnarray}\label{eq:ex3}
  \dot X=-(A+B_u)X\;,
\end{eqnarray}
where $A:=\Gamma_0+H_z$, $B:=uH_y$, and
\begin{equation}
\Gamma_0 := \gamma\,\diag(1,1,2)
\end{equation}
with $u\in\mathbb R$ and $\gamma\geq 0$.
So for approximating the corresponding Lie wedge, we take the first
step to be
\begin{equation}
\fw_1 := \mathbb R H_y \oplus (-\fc_1)\;,
\end{equation}
with $\fc_1:=\mathbb R_0^+ (H_z + \Gamma_0)$ and edge given by the
span of $H_y$ --- the control \/`Hamiltonian\/'.
Again, in step (2) we identify the conjugation to be brought about by
$\Adr_{\exp E(\mathfrak{w})}$ acting on the drift terms such as
to give in step (3) the set
\begin{eqnarray}
\fc_3 = \mathbb R_0^+ \{R_y(\theta)(\Gamma_0+H_z)R_y(\theta)^\top\;|\; \theta\in \mathbb R\}
\end{eqnarray}
with the $\fk$"~component brought about by the conjugated drift
\begin{equation}
R_y(\theta) H_z R_y(\theta)^\top = \cos(\theta)H_z+\sin(\theta)H_x
\end{equation}
and the $\fp$-component reading
\begin{equation}\label{eqn:p-Ex3}
\begin{split}
R_y &(\theta) \Gamma_0 R_y(\theta)^\top  \;=\;
\gamma \left[\begin{smallmatrix} 1+\sin^2(\theta) && 0 && \sin(\theta)\cos(\theta)\\
               0 && 1 && 0\\
               \sin(\theta)\cos(\theta) && 0 &&1+\cos^2(\theta)\end{smallmatrix}\right] \\[2mm]
& = \Gamma_0 + \tfrac{\gamma}{2} \Delta_{13}
 -\tfrac{\gamma}{2}\big(\cos(2\theta)\Delta_{13} -\sin(2\theta) p_y\big) \\[2mm]
  &=\tfrac{11+\cos(2\theta)}{12}\Gamma_0 + \tfrac{\gamma}{2}\big((1-\cos(2\theta))\Delta+\sin(2\theta)p_y\big)\;,\\
\end{split}
\end{equation}
where the matrices $\Delta_{ij}$ and $p_y$ are defined as in
Tab.~\ref{tab:comm-tab}, and, for the sake of orthogonality,
$\Delta:=\tfrac{2}{9}\unity+\tfrac{1}{18}\Delta_{12}+\tfrac{8}{9}\Delta_{13}$.

\medskip
Therefore, the Lie wedge can be expanded within the five-dimensional
Hilbert space
$\mathcal{H}:=  \mbox{span}\{H_x, H_z, p_y, \Delta, \Gamma_0\}$
and the final inner approximation to the Lie wedge takes the form
\begin{equation}\label{eqn:Ex3-cone-a}
\fw_0:=\mathbb R H_y \oplus - \fc_0 \;,
\end{equation}
where $\fc_0$ is parameterised (again in the short-hand of
\eqref{eqn:short-hand-scalprod}) as
\begin{equation}\label{eqn:Ex3-cone}
\fc_0:=\R{}_0^+\conv\Big\{
 \left[\begin{smallmatrix}
 2\sin(\theta)\\2\cos(\theta)\\ \gamma\sin(2\theta) \\ \gamma(1-\cos(2\theta))\\[.5mm] (11+\cos(2\theta))/6
\end{smallmatrix}\right]
\cdot
\left[\begin{smallmatrix}
H_x\\H_z\\p_y\\ \Delta\\[.5mm] \Gamma_0
\end{smallmatrix}\right] \;\Big|\; \theta\in \R{} \Big\}.
\end{equation}
It is shown in Fig.~\ref{fig:Ex3}. --- As in Example~2,
letting $\Adr_{\exp E(\mathfrak{w})}$ act on the drift terms
adds no further elements to the edge of the wedge, so one gets:

\medskip
\begin{proposition}\label{example3weakglobality}
The set $\fw_0 = \mathbb{R}H_y\oplus(-\fc_0)$ is
the {\em global} Lie wedge to control system \eqref{eq:ex3}.
\end{proposition}

\medskip
\begin{proof}
The Lie wedge property of $\fw_0$ can be derived as in Example 1.
Then the globality of $\fw_0$ follows again from Theorem
\ref{thm:globality-3}.
\end{proof}

\begin{figure}
\begin{center}
\includegraphics[scale=0.39]{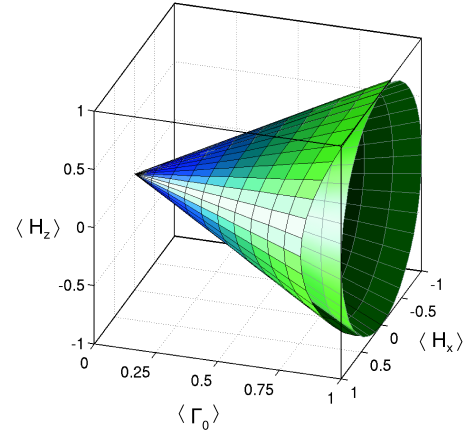}
\end{center}
\caption{\label{fig:Ex3} (Colour online) The Lie wedge of the more general
system in {\bf Example~3} is five-dimensional. Here the surface
of the convex cone $\fc_0$ generated by $e^{\theta H_y}(\Gamma_0+H_z)e^{-\theta H_y}$
(see \eqref{eqn:Ex3-cone} with $\lambda=1$)
is projected into the subspace $\mbox{span}\{H_x,H_z,\Gamma_0\}$.
Since $[\Gamma_0,H_y]\neq 0$ it deviates from the rotational symmetry
of Example~2 given in the inner part for comparison, cp Fig.~\ref{fig:Ex2}(a).
The figure scales with $\lambda\to\infty$.
The prism-shaped projection into the subspace $\mbox{span}\{H_y,H_z,\Gamma_0\}$ (not shown) is similar
to that of Example~2 already given in Fig.~\ref{fig:Ex2}(b).
}
\end{figure}


\medskip
\noindent
Generalising the relaxation operator in Example~3 to
$\Gamma_0 := \gamma\,\diag(a,b,c)$
with $a,b,c,\gamma\geq 0$ results in a generalised $\fp$-component
replacing \eqref{eqn:p-Ex3} by
\begin{equation*}
\begin{split}
&R_y (\theta) \Gamma_0 R_y(\theta)^\top  \;=\;
 \gamma\left[\begin{smallmatrix} a+(c-a)\sin^2(\theta) && 0 && (c-a)\sin(\theta)\cos(\theta)\\
               0 && b && 0\\
               (c-a)\sin(\theta)\cos(\theta) && 0 &&c-(c-a)\sin^2(\theta)\end{smallmatrix}\right] \\[2mm]
& =  \Gamma_0 + \gamma(c-a) \sin^2(\theta) \Delta_{13}
        +\gamma(c-a) \sin(\theta)\cos(\theta) p_y \\[2mm]
& = \Gamma_0 + \tfrac{\gamma(c-a)}{2} \Delta_{13}
 -\tfrac{\gamma(c-a)}{2}\big(\cos(2\theta)\Delta_{13} -\sin(2\theta) p_y\big)
\end{split}
\end{equation*}
This leads to a cone $\fc_0$ for \eqref{eqn:Ex3-cone-a} that keeps the structure of
\eqref{eqn:Ex3-cone} in a slightly more general form,
where Example~2 is readily reproduced by $c=a$, while Example~3
follows for $a=b=1$ and $c=2$. 

\medskip
The Lie wedge in Example~3 and its generalised form treated above
do {\em not} take the form of a Lie {\em semialgebra} either.
Choose $B:= H_y$ from the wedge $\fw_0$ of \eqref{eqn:Ex3-cone-a}
and recall $A:=\Gamma_0 + H_z$.
Then the BCH product 
\begin{equation}
A\star B = \Gamma_0 + H_y + H_z - \tfrac{1}{2}(H_x + p_y)\, +\, \dots
\end{equation}
leads outside the Lie wedge of \eqref{eqn:Ex3-cone-a}
since the component 
$H_x + p_y$ is {\em not} within\footnote{This would require the equality $2\,\sin(\theta)=\sin(2\theta)$
to hold for non-trivial~$\theta$ beyond its actual solutions $\theta=0\negthickspace\mod\pi$.}
the cone \eqref{eqn:Ex3-cone}.

\medskip
As pointed out already, in this section, we have chosen a
representation in $\mathbb R^3$ in order to visualise the
Hamiltonian parts of the respective quantum dynamics by
$SO(3)$-rotations. In quantum mechanics, 
this picture can be recovered in the so-called
{\em coherence-vector representation} \cite{AlickiLendi87}.
Therefore, when taking an explicit spin"~$\tfrac{1}{2}$
representation of $\adr_{\su(2)}$ 
in the following chapter,
the key results obtained here in
Examples~1 through~3
will show up again.
\section{Open Single-Qubit Quantum Systems}\label{sec:q-channels}

In this section, we analyze the standard single-qubit unital quantum systems 
beyond their purely dissipative evolution by allowing for Hamiltonian
drifts and controls.
In view of steering open quantum systems, this is an important generalisation.

\subsection{Markovian Master Equation in Qubit Representation}

Based on the Pauli matrices
\begin{equation}\label{eqn:Paulis}
\sigma_x:=
\left[\begin{matrix} 0 & 1\\ 1 &0 \end{matrix}\right],\quad
\sigma_y:=
\left[\begin{matrix} 0 & -i\\ i &0 \end{matrix}\right],\quad
\sigma_z:=
\left[\begin{matrix} 1 & 0\\ 0 & -1 \end{matrix}\right]
\end{equation}
in this section we deliberately depart from the previous notation by using
the explicit spin-$\tfrac{1}{2}$ adjoint representation carrying the
spin-quantum number $j=\tfrac{1}{2}$ as prefactor in
given by 
\begin{equation}\label{eqn:def-ad-sigma}
\hat{\sigma}_\nu:=
\tfrac{1}{2} (\unity_2\otimes\sigma_\nu - \sigma_\nu^\top\otimes\unity_2)\;,
\end{equation}
for $\nu\in\{x,y,z\}$, where $\otimes$ denotes
the Kronecker product of matrices.
One easily recovers the $\su(2)$ commutation relations 
\begin{equation}
[i\,\hat{\sigma}_p,i\,\hat{\sigma}_q]= - \varepsilon_{pqr} \;i\,\hat{\sigma}_r
\end{equation}
to convince oneself of
$\widehat{\adr}_{\su(2)} := \expt{i\hat{\sigma}_x, i\hat{\sigma}_y,i\hat{\sigma}_z}_{\sf Lie}
\iso\adr_{\su(2)}$.
Here and henceforth we use $\varepsilon_{pqr}$ to discriminate even and odd permutations
of $(x,y,z)$ by their signs, i.e.\
$\varepsilon_{pqr}=+1$ if $(p,q,r)$ is an {\em even} permutation
of $(x,y,z)$, while $\varepsilon_{pqr}=-1$ for an {\em odd} permutation.

Thus for 
a single open qubit system in the above representation
the controlled master equation \eqref{eqn:sigma} or rather
its group lift \eqref{eqn:SIGMA} takes the explicit
form
\begin{equation}
\dot X(t) =
-\Big(i\big(\hat{H}_d + \sum_j u_j\hat{H}_j\big)
+ \hat{\Gamma}_L\Big)\, X(t)\,.
\end{equation}
Here, $X(t)$ may be a density operator regarded (via the so-called
$\vec$-representation\footnote{Note that the $\vec$-representation
as well as the vector of coherence notation just provide different
\/`coordinates\/' for the abstract master equation \eqref{eqn:sigma}.
\cite{HJ2}) as an element in $\C{}^4$} or a qubit quantum channel represented in
$GL(4,\mathbb C)$. Moreover, $\hat{H}_d$ and $\hat{H}_j$ are
in general of the from
$\hat{H}_d := 
(\unity_2\otimes H_d - H_d^\top\otimes\unity_2)$
and similar $\hat{H}_j$. 
To ensure complete positivity, the relaxation term $\hat{\Gamma}_L$
shall be again of Lindblad-Kossakowski form
which for the standard unital single-qubit systems (with $V_k$ Hermitian)
simply reads 
$\hat{\Gamma}_L= 2\sum_k \gamma_k\;\hat{\sigma}_k^2$
to give the nicely structured generator
\begin{equation}\label{eqn:Lindblad-Lu}
\mathcal L_u =
i\big(\hat{H}_d + \sum_j u_j\hat{H}_j\big)
+\; 2 \negthickspace\negthickspace\sum_{k\in \{x,y,z\}}\negthickspace\negthickspace \gamma_k\;\hat{\sigma}_k^2\;.
\end{equation}
The generator is of this form because the $i\hat{H}$ terms are in the $\fk$-part
of the Cartan decomposition of $\gl(4,\mathbb C)$ into
skew-Hermitian ($\fk$) and Hermitian ($\fp$) matrices,
whereas the $\hat{\sigma}_k^2$ terms are in the $\fp$-part.


\subsection[Fully H-Controllable Channels]
{Single-Qubit Systems Satisfying Condition (H)}

Here we consider the class of fully Hamiltonian controllable
unital single-qubit systems whose dissipation is governed by
a \emph{single} Lindblad operator $\hat{\sigma}_k^2$ 
for some $k\in\{x,y,z\}$
i.e.~two of the three prefactors $\gamma_x,\gamma_y,\gamma_z$
have to vanish.

Similar to Example~1 of Sec.~\ref{sec:geoR3}, choose the controls
$\hat{\sigma}_x$ and $\hat{\sigma}_y$ to
see that such a system fulfills condition (H), since
$\expt{i\hat{\sigma}_x,i\hat{\sigma}_y}_{\sf Lie}=\widehat{\adr}_{\su(2)}$.
Then it is actually immaterial which single Pauli matrix is chosen
as the Lindblad operator $\hat{\sigma}^2_k$, because all of the
Pauli matrices are unitarily equivalent.
So without loss of generality, one may choose $k=z$,
i.e. $\gamma_x=0$, $\gamma_y=0$, and $\gamma_z =: \gamma$.

Therefore the fully Hamiltonian controllable version of the
bit-flip, phase-flip, and bit-phase-flip channels are dynamically
equivalent in as much as they have (up to unitary
equivalence) a common global Lie wedge
\begin{eqnarray}
  \fw_0 :=\widehat{\adr}_{\su(2)}\oplus-\fc_0\;,
\end{eqnarray}
where the cone $\fc_0$ is defined by
\begin{equation}\label{Eqn;fullHcontrolcone1qubit}
\fc_0:=
\R{}_0^+\mbox{conv}\left\{\hat{U}\,\hat{\sigma}_z^2\,\hat{U}^\dagger
\;\big|\;U\in SU(2) \right\} \\[2mm] \quad
\end{equation}
with 
\begin{equation}\label{eqn:U-superop}
\hat{U} := \bar U \otimes U\;. 
\end{equation}
Clearly, the
wedge $\fw_0$ is \emph{global} by Corollary~\ref{HcontrolGlobalWedge}
or, alternatively, by Theorem \ref{thm:globality-3} and its edge $E(\fw_0)$ is
given by the Lie subalgebra $\widehat{\adr}_{\su(2)}$.
The above Lie wedge is isomorphic to the one in
Example~1 of Sec.~\ref{sec:geoR3} for the particular
choice that $\Gamma_0 = \diag(1,1,0)$.

\subsection[Controllable Channels Satisfying Condition (WH)~I: One Lindblad Operator]
{Single-Qubit Systems Satisfying Condition (WH): One Lindblad Operator}

Here we discuss an important class of standard single-qubit systems
which are particularly simple in three regards
\begin{enumerate}
\item[(i)] their dissipative term is governed by a single
Lindblad operator,
$\Gamma_0 :=
2\gamma\hat{\sigma}^2_k$
for some $k\in\{x,y,z\}$;
\item[(ii)] their switchable Hamiltonian control is brought
about by a single Hamiltonian $\hat{\sigma}_c$
for some $c\in\{x,y,z\}$;
\item[(iii)] their non-switchable Hamiltonian drift is
$\hat{\sigma}_d$ for some $d\in\{x,y,z\}$.
\end{enumerate}

Applying the algorithm for the inner approximation of the Lie wedge,
we get in step (1)
\begin{equation}
\fw^c_{dk}(1) :=
i\; \mathbb R \hat{\sigma}_c \oplus\,
-\mathbb R_0^+ \big(i\hat{\sigma}_d+2\gamma\hat{\sigma}_k^2\big)\;,
\end{equation}
where again we note the separation by $\fk$-$\fp$ components.
In step (2) we identify the span generated by the control $i\hat{\sigma}_c$
as the edge $E(\fw)$ of the wedge.
So the conjugation has to be by the control subgroup, i.e.\ by
$
e^{-i2\theta\hat{\sigma}_c} = 
e^{+i\theta\sigma_c^\top}\otimes e^{-i\theta\sigma_c}$.
Thus in step (3)
one obtains as $\fk$"~component of the conjugated drift
\begin{equation}\label{Eqn:Kcomponenttrig}
\begin{split}
K_d^c(\theta) := &\quad e^{-i\theta\hat\sigma_c}(i\hat\sigma_d)e^{i\theta\hat\sigma_c}\\
   = &\begin{cases} i\,\hat\sigma_d &\text{for $c=d$}\quad\\
                   i\,\cos(\theta)\hat\sigma_d+ i\,\varepsilon_{cdq}\sin(\theta)\hat\sigma_q &\text{else}
     \end{cases}
\end{split}
\end{equation}
and as $\fp$-component
\begin{equation}\label{Eqn:Pcomponenttrig}
\begin{split}
P_k^c(\theta) := &\quad e^{-i\theta\hat\sigma_c}(2\gamma\,\hat\sigma_k^2)e^{i\theta\hat\sigma_c}\\[1mm]
   = &\begin{cases} 2\gamma\,\hat\sigma_k^2 &\text{for $c=k$}\quad\\[1mm]
                   2\gamma\,\big(\cos(\theta)\hat\sigma_k+ \,\varepsilon_{ckr}\sin(\theta)\hat\sigma_r\big)^2 &\text{else}\;.
     \end{cases}
\end{split}
\end{equation}
The last expression (for $c\neq k$) can be further resolved 
using 
the anticommutator $\left\{A,B\right\}_+:=AB+BA$
\begin{equation}\label{eqn:qubitchancor2}
\begin{split}
P^c_k(\theta)  
&= 2\gamma\left[\begin{smallmatrix}\cos^2(\theta)\\
                                        \sin^2(\theta)\\
                                        \cos(\theta)\sin(\theta)\end{smallmatrix}\right]\cdot
               \left[\begin{smallmatrix}\hat\sigma_k^2\\ \hat\sigma_r^2\\ \varepsilon_{ckr}\{\hat\sigma_k,\hat\sigma_r\}_+ \end{smallmatrix}\right]\,.\\[2mm]
&= \tfrac{\gamma}{2}\left[\begin{smallmatrix}2\\
                                        1+\cos(2\theta)\\
                                        1-\cos(2\theta)\\
                                        \sin(2\theta)\end{smallmatrix}\right]\cdot
               \left[\begin{smallmatrix}\unity\\-(\sigma^\top_k\otimes\sigma_k)\\
                                                -(\sigma^\top_r\otimes\sigma_r)\\
                                                -\varepsilon_{ckr} ((\sigma^\top_k\otimes\sigma_r)+(\sigma^\top_r\otimes\sigma_k))\end{smallmatrix}\right]\;,\\
\end{split}
\end{equation}
where the latter identity gives a decomposition into mutually orthogonal Pauli-basis elements.

To summarize, if the control Hamiltonian neither commutes with the Hamiltonian
part nor with the dissipative part of the drift, one obtains
in terms of the above $K_d^c(\theta)$ and $P_k^c(\theta)$
\begin{eqnarray}\label{eqn:QubitCone}
\fc_{dk}^c:= \mathbb{R}_0^+ \conv\{K_d^c(\theta)+P_k^c(\theta)\,|\,\theta\in\R{}\,\}
\end{eqnarray}
However, if $[\hat\sigma_c,\Gamma_0]=0$, then the convex cone in equation \eqref{eqn:QubitCone}
simplifies by $P_k^c(\theta)=\Gamma_0$ to
\begin{equation}\label{eqn:CommQubitCone}
\fc_{dk}^c 
=\R{}_0^+\conv\Big\{\left[\begin{smallmatrix}\cos(\theta)\\\sin(\theta)\\1\end{smallmatrix}\right]\cdot
 \left[\begin{smallmatrix}i\hat\sigma_d \\ i\varepsilon_{cdq}\hat\sigma_q\\\Gamma_0\end{smallmatrix}\right]
        \;\Big|\; \theta \in \R{} \Big\}\quad
\end{equation}
in entire analogy to Example~2 of Sec.~\ref{sec:geoR3}.

The final Lie wedge admits the orthogonal decomposition
\begin{equation}
\fw^c_{dk} := i \mathbb R \,\hat\sigma_c \oplus - \fc^c_{dk}
\end{equation}
and moreover by Theorem \ref{thm:globality-3} (or alternatively by
Corollary \ref{HHLglobalcor}) it is {\em global}.
For $\gamma>0$, the edge $E(\fw)=\expt{i\,\hat\sigma_c}$ is again
the span generated by the control, yet it flips into the full algebra
$E(\fw)=\widehat{\adr}_{\su(2)}$ in the limit $\gamma = 0$.

The relation to Examples~2 and 3 of Sec.~\ref{sec:geoR3} is obvious:
Let a unital qubit system satisfy the (WH)-condition
and have a dissipative Lindbladian
$\Gamma_0:=2\gamma\hat\sigma_k^2$ induced by a single Lindblad
operator $V_k = \sigma_k$. If $[\hat\sigma_c,\Gamma_0]\neq 0$, 
one arrives at a situation resembling Example~3, whereas if
$[\hat\sigma_c,\Gamma_0] = 0$, one obtains a result analogous to Example~2.

\subsection*{Application: Bit-Flip and Phase-Flip Channels}
Also the relation to standard {\em unital} qubit channels is immediate:
Note that in the bit-flip channel the noise is generated by
$\hat\sigma_x^2$, while it is
$\hat\sigma_y^2$ in the bit-phase-flip channel and $\hat\sigma_z^2$ in
the phase-flip channel, see Tab.~\ref{tab:q-channels}.
In the absence of any {\em coherent} drift or control brought about
by the respective Hamiltonians 
$H_d=\hat \sigma_d$ or $H_j=\hat \sigma_j$, the Kraus representations
are standard. By allowing for drifts and controls, the Kraus rank $K$
of the channel usually increases to $K$=$4$
with exception of a single $\hat \sigma_d$ or
$\hat \sigma_j$ commuting with the single
Lindblad operator {$\hat \sigma_k^2$} keeping $K$=$2$.
Also the time dependences become more involved. Hence explicit results will be given elsewhere.

Under full H"~controllability, the Lie wedges of all the three
channels become equivalent as the Pauli
matrices and thus the corresponding noise generators are unitarily similar.

In contrast, for the case satisfying the (WH)"~condition, assume
a control system with a Hamiltonian drift term governed by $\hat\sigma_z$.
Upon including relaxation, now there are two different scenarios:
if the control Hamiltonian (indexed by $c\in\{x,y,z\}$)
commutes with the noise generator (indexed by $k\in\{x,y,z\}$), one
finds a situation as in Example~2 and \eqref{eqn:CommQubitCone}, otherwise the
scenario is more general as in \eqref{eqn:QubitCone}.

\begin{table*}[Ht!]
\caption{\label{tab:q-channels}
Controlled Single-Qubit Channels and Their Lie Wedges}
\vspace{-2mm}
\begin{center}
\begin{tabular}{l|ll ll}
\hline\hline\\
{\small Channel} & {\small Primary$^*$ Lindblad Operators} & {\small Primary$^*$ Kraus Operators}
        & \multicolumn{2}{c}{\small ------------------------ Lie Wedges ------------------------- \phantom{XXXX}}\\
 & &  & case satisfying (WH)-condition & H-controllable case\\[1mm]
\hline\\[-0.8mm]
Bit Flip &
$V_1=\sqrt{a_{11}}\;\sigma_x$
&
$E_1=\sqrt{r_{11}}\;\sigma_x$
&
$\mathfrak{w}_{dx}^c=\langle
i\hat\sigma_c\rangle\oplus-\mathfrak{c}_{dx}^c$&
$\mathfrak{w}=\widehat{\adr}_{\mathfrak{su}(2)}\oplus-\mathfrak{c}_0$\\[1mm]
&&
$E_0=\sqrt{q_{11}}\;\unity$  & [see Eqns.~(\ref{eqn:QubitCone},\ref{eqn:CommQubitCone})]
                             & [see Eqn.~\eqref{Eqn;fullHcontrolcone1qubit}] \\[3mm]
Phase Flip &
$V_1=\sqrt{a_{22}}\;\sigma_z$ 
&
$E_1=\sqrt{r_{22}}\;\sigma_z$
&
$\mathfrak{w}_{dz}^c=\langle
i\hat\sigma_c\rangle\oplus-\mathfrak{c}_{dz}^c$&
---same as above---\\[1mm]
&&
$E_0=\sqrt{q_{22}}\;\unity$ & [see Eqns.~(\ref{eqn:QubitCone},\ref{eqn:CommQubitCone})]\\[3mm]
Bit-Phase Flip &
$V_1=\sqrt{a_{33}}\;\sigma_y$
&
$E_1=\sqrt{r_{33}}\;\sigma_y$
&
$\mathfrak{w}_{dy}^c=\langle
i\hat\sigma_c\rangle\oplus-\mathfrak{c}_{dy}^c$&
---same as above---\\[1mm]
&&
$E_0=\sqrt{q_{33}}\;\unity$ & [see Eqns.~(\ref{eqn:QubitCone},\ref{eqn:CommQubitCone})]\\[3mm]
\hline\\
Depolarizing & 
$V_1=\sqrt{a_{11}}\;\sigma_x$
&
$E_1=\sqrt{r_1}\;\sigma_x$
&
$\mathfrak{w}_{d,xyz}^c=\langle
i\hat\sigma_c\rangle\oplus-\mathfrak{c}_{d,xyz}^c$&
$\fw=\widehat{\adr}_{\su(2)}\oplus -\fc_{xyz}$ \\[1mm]
&$V_2=\sqrt{a_{22}}\;\sigma_y$
&
$E_2=\sqrt{r_2}\;\sigma_y$
& \multicolumn{1}{l}{[see Eqn.~\eqref{Eqn:Kcomponenttrig} and Eqns.~(\ref{eqn:Pkkk},\ref{eqn:Pkkk-simple})]}
& [see Eqns.~(\ref{eqn:H-cont-depol-1},\ref{eqn:H-cont-depol-2})]\\[1mm]
&$V_3=\sqrt{a_{33}}\;\sigma_z$
&
$E_3=\sqrt{r_3}\;\sigma_z$
&&\\[1mm]
&&
$E_0=\sqrt{r_0}\;\unity$&&\\[2mm]
\multicolumn{5}{l}{\phantom{XX}}\\[-3mm]
\hline\hline\\
\multicolumn{5}{l}{$^*$) Primary operators are for purely dissipative time evolutions
(no Hamiltonian drift no control). 
Then the time dependence of the}\\
\multicolumn{5}{l}{Kraus operators roots in the GKS matrix 
                                        $\{a_{ii}\}_{i=1}^3$. Define:
$ \lambda_1:=a_{22}+a_{33}, \lambda_2:=a_{22}-a_{33}, \lambda_3:=a_{11}+a_{22}$, and thereby
$q_{ii}:=\tfrac{1}{2}(1+e^{-a_{ii}t}), $}\\[1mm]
\multicolumn{5}{l}{$r_{ii}:=\tfrac{1}{2}(1-e^{-a_{ii}t})$ and $r_0:=\tfrac{1}{4}(1+e^{-\lambda_1 t}+e^{-\lambda_2 t}+e^{-\lambda_3 t}),\;
  r_1:=\tfrac{1}{4}(1-e^{-\lambda_1 t}+e^{-\lambda_2 t}-e^{-\lambda_3 t}),\;
  r_2:=\tfrac{1}{4}(1+e^{-\lambda_1 t}-e^{-\lambda_2 t}-e^{-\lambda_3 t})$,}\\[1mm]
\multicolumn{5}{l}{$r_3:=\tfrac{1}{4}(1-e^{-\lambda_1 t}-e^{-\lambda_2 t}+e^{-\lambda_3 t})$. ---
                Under Hamiltonian drift and control the Kraus-rank gets $K=4$ except for one single control $H_c$ or drift $H_d$}\\
\multicolumn{5}{l}{that commutes with the only Lindblad operator $V_k$: in this case the Kraus rank is $K=2$. Time-dependences are involved
                and will be given elsewhere.}\\
\multicolumn{5}{l}{The Pauli matrices $\sigma_\nu$ are defined in Eqn.~\eqref{eqn:Paulis}.
}\\
\end{tabular}
\end{center}
\end{table*}

\subsection[Controllable Channels Satisfying Condition (WH)~II: Several Lindblad Operators]
{Single-Qubit Systems Satisfying Condition (WH): Several Lindblad Operators}

Consider a unital qubit system satisfying the (WH)-condition and whose
Lindbladian $\Gamma_0$ is generated by $\ell = 2$ or $\ell = 3$ different
Lindblad operators $\hat\sigma_k^2$.
Then one obtains the following generalisations of the
symmetric component $P_k^c(\theta)\in\fc_{dk}^c$.\\[3mm]
For $\ell=2$ and  $\sigma_c\perp\sigma_k$, $\sigma_c=\sigma_{k'}$
    \begin{equation}
\begin{split}
\label{eqn:2Lind1comm}
P^c_{kk'}(\theta)
&= 2\left[\begin{smallmatrix}\gamma'\\\gamma\cos^2(\theta)\\
                                        \gamma\sin^2(\theta)\\
                                        \gamma\cos(\theta)\sin(\theta)\end{smallmatrix}\right]\cdot
               \left[\begin{smallmatrix}\hat\sigma^2_{k'}\\ \hat\sigma^2_k\\ \hat\sigma_r^2\\
                                        \varepsilon_{ckr}\{\hat\sigma_k,\hat\sigma_r\}_+ \end{smallmatrix}\right]\\[2mm]
= \tfrac{1}{2}&\left[\begin{smallmatrix} \gamma' \\ 2 (\gamma+\gamma')\\
                                        \gamma(1+\cos(2\theta))\\
                                        \gamma(1-\cos(2\theta))\\
                                        \gamma\sin(2\theta)\end{smallmatrix}\right]\cdot
               \left[\begin{smallmatrix} -(\sigma^\top_{k'}\otimes\sigma_{k'})\\\unity\\
                                        -(\sigma^\top_k\otimes\sigma_k)\\
                                        -(\sigma^\top_r\otimes\sigma_r)\\
                                        -\varepsilon_{ckr} ((\sigma^\top_r\otimes\sigma_k)+(\sigma^\top_k\otimes\sigma_r))\end{smallmatrix}\right]. \\[2mm]
\end{split}
\end{equation}
while for $\ell=3$ and $\sigma_c\perp\sigma_k$, $\sigma_c\perp\sigma_{k'}$, $\sigma_c=\sigma_{k''}$
\begin{equation}\label{eqn:Pkkk}
\begin{split}
&P^c_{kk'k''}(\theta)
= 2\left[\begin{smallmatrix}\gamma''\\
                                        \gamma\cos^2(\theta)+\gamma'\sin^2(\theta)\\
                                        \gamma'\cos^2(\theta)+\gamma\sin^2(\theta)\\
                                        (\gamma-\gamma')\cos(\theta)\sin(\theta)\end{smallmatrix}\right]\cdot
               \left[\begin{smallmatrix}\hat\sigma^2_{k''} \\ \hat\sigma^2_k\\ \hat\sigma^2_{k'}\\ 
                \varepsilon_{ckk'}\{\hat\sigma_k,\hat\sigma_{k'}\}_+ \end{smallmatrix}\right]\\[2mm]
= \tfrac{1}{2}&\left[\begin{smallmatrix}\gamma''\\
                                        2(\gamma+\gamma'+\gamma'')\\
                                        \gamma+\gamma'+(\gamma-\gamma')\cos(2\theta)\\
                                        \gamma+\gamma'-(\gamma-\gamma')\cos(2\theta)\\
                                        (\gamma-\gamma')\sin(2\theta)\end{smallmatrix}\right]\cdot
               \left[\begin{smallmatrix}-(\sigma^\top_{k''}\otimes\sigma_{k''})\\
                                        \unity\\-(\sigma^\top_k\otimes\sigma_k)\\
                                                -(\sigma^\top_{k'}\otimes\sigma_{k'})\\
                                                -\varepsilon_{ckk'} ((\sigma^\top_{k'}\otimes\sigma_k)+(\sigma^\top_k\otimes\sigma_{k'}))\end{smallmatrix}\right].\\[2mm]
\end{split}
\end{equation}
which for $\gamma=\gamma'=\gamma''$ simplifies to
    \begin{eqnarray}\label{eqn:Pkkk-simple}
      P^c_{kk'k''}(\theta)=\Gamma_0=2\gamma(\hat\sigma^2_k+\hat\sigma^2_{k'}+\hat\sigma^2_{k''})\;.
    \end{eqnarray}
Note that \eqref{eqn:2Lind1comm} with $\gamma=\gamma'$ precisely
corresponds to Example~3 in Sec.~\ref{sec:geoR3}.

\subsection*{Application: Depolarising Channel}
Treating the depolarising channel also becomes immediate,
since one has three noise generators governed by all of $\hat\sigma_x,\hat\sigma_y$,
and $\hat\sigma_z$.
Thus the fully Hamiltonian controllable version of the
depolarising channel follows the bit-flip and phase-flip channels
in the structure of its global Lie wedge
\begin{eqnarray}\label{eqn:H-cont-depol-1}
  \fw_0 :=\widehat{\adr}_{\su(2)}\oplus-\fc_{xyz}\;,
\end{eqnarray}
where the cone $\fc_{xyz}$ now reads
\begin{equation}\label{eqn:H-cont-depol-2}
\fc_{xyz}:=\R{}_0^+\mbox{conv}\big\{\hat{U}(\gamma_x\hat\sigma_x^2 + \gamma_y\hat\sigma_y^2 %
        +\gamma_z\hat\sigma_z^2)\hat{U}^\dagger\;|\;U\in SU(2)\big\}
\end{equation}
with $\hat U$ of the form \eqref{eqn:U-superop}.
Again, the edge of the wedge is given by the entire algebra $E(\fw)=\widehat{\adr}_{\su(2)}$
and globality of the wedge follows by Theorem~\ref{thm:globality-3} or
Corollary~\ref{HcontrolGlobalWedge}. ---
Moreover, note that the Lie wedge in the fully Hamiltonian controllable depolarising channel with
{\em isotropic} noise takes the structure of a Lie semialgebra as (in the coherence-vector
representation) it corresponds to the special case of Example~1 in
Sec.~\ref{sec:geoR3}, where the relaxation
operator is a scalar multiple of the unity, $\Gamma_0=\lambda\cdot\unity$.
For {\em anisotropic} relaxation, however, this feature does not arise.

\medskip
If only condition (WH) is satisfied, there are two distinctions:
if the noise contributions are isotropic (i.e.\ with equal contribution by all the
Paulis through $\gamma_x=\gamma_y=\gamma_z$),
one finds a cone expressed by \eqref{Eqn:Kcomponenttrig} and \eqref{eqn:Pkkk-simple}.
However, in the generic anisotropic case, the cone can be expressed by
\eqref{Eqn:Kcomponenttrig} and \eqref{eqn:Pkkk}, see also Tab.~\ref{tab:q-channels}.

\section{Open Two-Qubit Quantum Systems}

In this section we extend the notions introduced in the previous
chapter to three types of two-qubit quantum systems with uncorrelated noise.
The two qubits will be denoted $\sA$ and $\sB$, respectively.
Moreover, we use the short-hands $\sigma_{\mu\nu}:=\sigma_\mu\otimes\sigma_\nu$
with $\mu,\nu\in\{x,y,z,1\}$, where $\sigma_1:=\unity$ as well as
the corresponding \/`commutator superoperators\/'
$\hat\sigma_{\mu\nu}:=\tfrac{1}{2}(\unity\otimes\sigma_{\mu\nu}-\sigma_{\mu\nu}^t\otimes\unity)$.

\subsection[Fully H"~Controllable Channels]{Fully H"~Controllable Two-Qubit Channels}
A fully Hamiltonian controllable two-qubit toy-model
system with switchable Ising-coupling is given by the master equation
\begin{eqnarray}\label{hcontrol2qubiteqn}
     \dot\rho= 
		-\big(i \sum_j u_j\hat\sigma_j + \Gamma_0\big)\rho\;
\end{eqnarray}
where $\hat\sigma_j\in\{\hat\sigma_{x1},\hat\sigma_{y1};\hat\sigma_{1x},\hat\sigma_{1y};\hat\sigma_{zz}\}$ 
are the Hamiltonian control terms with amplitudes $\left\{u_j\right\}_{j=1}^5\in\mathbb{R}$.

Since $\expt{i\hat\sigma_j\,|\,j=1,2,\dots,5}_{\sf Lie} =\widehat{\adr}_{\su(4)}$,
the edge of the wedge is $E(\fw)=\widehat{\adr}_{\su(4)}$.
Following the algorithm for an inner approximation of the
Lie wedge, step (1) thus gives
\begin{eqnarray}
  \fw_1:=\widehat{\adr}_{\mathfrak{su}(4)}\oplus-\mathbb{R}^+_0\Gamma_0\;.
\end{eqnarray}
Conjugating the dissipative component by the exponential map of the edge
and then taking the convex hull yields the convex cone
\begin{eqnarray}
  \fc_0:=\mbox{conv}\big\{\lambda\hat U\Gamma_0 \hat U^{\dagger}| \hat U:=\bar U \otimes U, U\in SU(4), \lambda\geq 0\big\}\;,
\end{eqnarray}
which is the two-qubit analogue of the cone in Eqn.~\eqref{Eqn;fullHcontrolcone1qubit}.
The resulting Lie wedge
\begin{eqnarray}
  \fw_0:=\widehat{\adr}_{\su(4)}\oplus-\fc_0
\end{eqnarray}
is global by Theorem~\ref{thm:globality-3}. 

\subsection[Controllable Channels Satisfying (H)"~Condition Locally and (WH)"~Condition Globally]
	{Two-Qubit Channels Satisfying the (H)"~Condition Locally and the (WH)"~Condition Globally}
By shifting the Ising coupling term from the set of switchable control Hamiltonians
into the (non-switchable) drift term, $\hat\sigma_d= \hat\sigma_{zz}$,
one obtains the realistic and actually widely occuring type of system
\begin{eqnarray}\label{eqn:HWHcontrol2qubit}
     \dot\rho= 
		-\big(i \hat\sigma_d + i\sum_j u_j\hat\sigma_j + \Gamma_0\big)\rho\;
\end{eqnarray}
where now one just has the local control terms 
$\hat\sigma_j\in\{\hat\sigma_{x1},\hat\sigma_{y1};\hat\sigma_{1x}, \hat\sigma_{1y}\}$.
Since $\expt{i\hat\sigma_{x1}, i\hat\sigma_{y1}}_{\sf Lie} =\widehat{\adr}_{\su_{\sA}(2)}\otimes\unity_{\sB}$,
whereas on the other hand $\expt{i\hat\sigma_{1x},i\hat\sigma_{1y}}_{\sf Lie} =\unity_{\sA}\otimes\widehat{\adr}_{\su_{\sB}(2)}$,
the edge of the wedge
\begin{equation}
E(\fw)=\widehat{\adr}_{\su_{\sA}(2)\widehat{\oplus}\su_{\sB}(2)}
\end{equation}
is in fact brought about by the Kronecker sum of local algebras
\begin{equation}
\su_{\sA}(2)\otimes\unity_{\sB} + \unity_{\sA}\otimes\su_{\sB}(2) =: \su_{\sA}(2)\widehat{\oplus}\su_{\sB}(2)
\end{equation}
forming the generator of the group of local unitary actions
\begin{equation}
\exp \big(\su_{\sA}(2)\widehat{\oplus}\su_{\sB}(2)\big) = SU_{\sA}(2)\otimes SU_{\sB}(2)\;.
\end{equation}

Remarkably, in this important class of open quantum-dynamical systems,
qubits \sA and \sB are {\em locally} (H)"~controllable, respectively,
while {\em globally} the system satisfies but the (WH)"~condition.

\medskip
The final Lie wedge in these systems reads as
\begin{eqnarray}
 \fw_{dk}^{2\oplus 2}=\widehat{\adr}_{\su_{\sA}(2)\widehat{\oplus}\su_{\sB}(2)}\oplus-\fc_{dk}^{2\oplus 2}
\end{eqnarray}
with the convex cone
\begin{eqnarray}
  \fc_{dk}^{2\oplus 2}:=\mathbb{R}^+_0 \conv\big\{K_{d}^{2\oplus 2}+P_{k}^{2\oplus 2}\big\}
\end{eqnarray}
being given in terms of the respective $\fk$ and $\fp$-components.
Here we use the short-hand of \eqref{eqn:U-superop} in the sense of
$\hat U_{2\otimes 2}:=\bar U_{2\otimes 2}\otimes U_{2\otimes 2}$ 
to arrive at
\begin{eqnarray}
K_d^{2\oplus 2}&:=&\{\hat U_{2\otimes 2}(i\hat \sigma_d)\hat U^\dagger_{2\otimes 2}\,|\,U_{2\otimes 2}\in SU(2)\otimes SU(2)\}\qquad\\
P_k^{2\oplus 2}&:=&\{\hat U_{2\otimes 2}(\Gamma_0)\hat U^\dagger_{2\otimes 2}\,|\,U_{2\otimes 2}\in SU(2)\otimes SU(2)\}\,.
\end{eqnarray}

As before, this immediately results from the initial wedge approximation by step (1)
\begin{eqnarray}
  \fw_1^{2\oplus 2}:=\widehat{\adr}_{\mathfrak{su}_{\sA}(2)\widehat{\oplus}\mathfrak{su}_{\sB}(2)}\oplus
  -\mathbb{R}^+_0(i\hat \sigma_d+\Gamma_0)
\end{eqnarray}
followed by conjugation with $\Adr_{\exp E(\fw)}=\Adr_{2\otimes 2}$ to give
\begin{equation}\label{eqn:2qubiKP}
  K_d^{2\oplus 2}+P_k^{2\oplus 2}
        := \mathcal O_{SU(2)\otimes SU(2)}\big(i\hat \sigma_d+\Gamma_0\big)\;.\quad
\end{equation}
Step (3) then takes the convex hull. ---
To show globality, let $\fw'$ denote the global
Lie wedge corresponding to the fully H"~controllable
system given in Eqn.~ \eqref{hcontrol2qubiteqn}. Then $\fw_{dk}^{2\oplus 2}\subset\fw'$, and it can be
shown that $\fw_{dk}^{2\oplus 2}$ satisfies the conditions of Corollary \ref{HHLglobalcor}
and therefore is {\em global}.

\subsection[Controllable Channels Satisfying Only (WH)"~Condition]
	{Two-Qubit Channels Satisfying Only the (WH)"~Condition}
In the final example of a two-qubit system, the independent local controls shall even
be limited to either $x$ or $y$-controls on the two qubits according to
\begin{eqnarray}\label{eqn:WHWHcontrol2qubit}
     \dot\rho 
            =-\big(i(\hat\sigma_d +  u_{\sA}\hat\sigma_{c1} + u_{\sB}\hat\sigma_{1c'})+\Gamma_0\big)\rho\;,
\end{eqnarray}
where now $\hat\sigma_d:=i\big(\hat\sigma_{z1} + \hat\sigma_{1z} + \hat\sigma_{zz}\big)$ and
$\hat\sigma_{c1}$ with a single $c\in\{x,y\}$ and likewise $\hat\sigma_{1c'}$ with
a single $c'\in\{x,y\}$ and $u_{\sA},u_{\sB}\in\mathbb{R}$.
Furthermore, assume the system undergoes {\em local uncorrelated noise}
in each of the two subsystems in the sense that the Lindblad operators are of local form
\begin{eqnarray}
V_k &\in &\{\sigma_{k1}\;|\; k\in\{x,y,z\}\}\\
V_{k'} &\in &\{\hat\sigma_{1k'}\;|\; k'\in\{x,y,z\}\}\;,
\end{eqnarray}
where $k$ and $k'$ are chosen independently $k,k'\in\{x,y,z\}$ so that
in the convention of \eqref{eqn:def-ad-sigma} one finds
\begin{equation}
\Gamma_0\;:=\;2\gamma\hat \sigma_{k1}^2\; +\; 2\gamma'\hat \sigma^2_{1k'}\;.
\end{equation}
This system satisfies but the (WH)"~condition both locally and globally,
the latter following from
\begin{equation}
\expt{i\hat\sigma_{c1}, i\hat\sigma_{1c'}, i\hat\sigma_d}_{\sf Lie} = \widehat{\adr}_{\su(4)}\;.
\end{equation}

\bigskip
The Lie wedge is given by
\begin{eqnarray}
  \fw^{cc'}_{kk'}:=\expt{i\hat \sigma_c} + \expt{i\hat \sigma_{c'}} \oplus - \fc^{cc'}_{kk'}\;,
\end{eqnarray}
where the two-dimensional edge of the wedge is generated by the rays
$\expt{i\hat \sigma_{c1}}$, $\expt{i\hat\sigma_{1c'}}$ and the cone
\begin{equation}
\begin{split}
  \fc^{cc'}_{kk'}:=\mathbb{R}_0^+ \conv\big\{& K^c(\theta)+K^{c'}(\theta') +K^{cc'}(\theta,\theta')\\
                        &+ P_k^c(\theta)+P^{c'}_{k'}(\theta')\,|\,\theta,\theta'\in\mathbb R\,\big\}
\end{split}
\end{equation}
is given in terms of the $\fk$- and $\fp$"~components
(setting $\theta:=u_{\sA}$ and $\theta':=u_{\sB}$
and using the relations in \eqref{Eqn:Kcomponenttrig}) as
\begin{equation}\label{eqn:two-qubit-K123}
K^c(\theta) + K^{c'}(\theta') + K^{cc'}(\theta,\theta')  
        = \left[\begin{smallmatrix}\cos(\theta)\\\sin(\theta)\\\cos(\theta')\\\sin(\theta')\\ \cos(\theta)\cos(\theta') \\
        \cos(\theta)\sin(\theta')\\ \cos(\theta')\sin(\theta)\\ \sin(\theta)\sin(\theta') \end{smallmatrix}\right]
                \cdot i
        \left[\begin{smallmatrix}\hat \sigma_{z1}\\ \varepsilon_{czq}\hat \sigma_{q1}\\
                \hat \sigma_{1z}\\ \varepsilon_{c'zq'} \hat \sigma_{1q'}\\
                \hat \sigma_{zz}\\ \varepsilon_{c'zq'} \hat \sigma_{zq'}\\
                \varepsilon_{czq} \hat \sigma_{qz}\\
                \hat \sigma_{qq'} \end{smallmatrix}\right]\qquad
\end{equation}
and (as in \eqref{eqn:qubitchancor2})
\begin{equation}\label{eqn:two-qubit-P1}
\begin{split}
P^c_k&(\theta)
 = 2\gamma\left[\begin{smallmatrix}\cos^2(\theta)\\
                                        \sin^2(\theta)\\
                                        \cos(\theta)\sin(\theta)\end{smallmatrix}\right]\cdot
               \left[\begin{smallmatrix}\hat \sigma^2_{k1}\\
                                        \hat \sigma^2_{r1}\\
                        \varepsilon_{ckr}\{\hat \sigma_{k1},\hat \sigma_{r1}\}_+ \end{smallmatrix}\right]\\[2mm]
&= \tfrac{\gamma}{2}\left[\begin{smallmatrix}2\\
                                        1+\cos(2\theta)\\
                                        1-\cos(2\theta)\\
                                        \sin(2\theta)\end{smallmatrix}\right]\cdot
               \left[\begin{smallmatrix}\unity\\-(\sigma^\top_k\otimes\sigma_k)\otimes\unity_{\sB}\\
                                                -(\sigma^\top_r\otimes\sigma_r)\otimes\unity_{\sB}\\
                                                -\varepsilon_{ckr} ((\sigma^\top_k\otimes\sigma_r)+(\sigma^\top_r\otimes\sigma_k))\otimes\unity_{\sB}\end{smallmatrix}\right]\\[2mm]
\end{split}
\end{equation}
as well as
\begin{equation}\label{eqn:two-qubit-P2}
\begin{split}
P^{c'}_{k'}&(\theta')
 = 2\gamma'\left[\begin{smallmatrix}\cos^2(\theta')\\
                                        \sin^2(\theta')\\
                                        \cos(\theta')\sin(\theta')\end{smallmatrix}\right]\cdot
               \left[\begin{smallmatrix}\hat \sigma_{1k'}^2\\
                                        \hat \sigma_{1r'}^2\\
                        \varepsilon_{c'k'r'}\{\hat \sigma_{1k'}^2,
                                              \hat \sigma_{1r'}^2\}_+ \end{smallmatrix}\right]\\[2mm]
&= \tfrac{\gamma}{2}\left[\begin{smallmatrix}2\\
                                        1+\cos(2\theta')\\
                                        1-\cos(2\theta')\\
                                        \sin(2\theta')\end{smallmatrix}\right]\cdot
               \left[\begin{smallmatrix}\unity\\-\unity_{\sA}\otimes(\sigma^\top_k\otimes\sigma_k)\\
                                                -\unity_{\sA}\otimes(\sigma^\top_r\otimes\sigma_r)\\
                        -\varepsilon_{c'k'r'} \unity_{\sA}\otimes((\sigma^\top_{k'}\otimes\sigma_{r'})+(\sigma^\top_{r'}\otimes\sigma_{k'}))
                \end{smallmatrix}\right].\\[2mm]
\end{split}
\end{equation}

To see this, observe that by step (1), the initial wedge approximation is given by
\begin{equation}
  \fw_1:= \expt{i \hat \sigma_c}+\expt{i\hat \sigma_{c'}}
        \oplus-\mathbb R^+_0(i\hat \sigma_d+2\gamma\hat \sigma_k^2+2\gamma'\hat \sigma^2_{k'})\,,\\[2mm]
\end{equation}
which has to be conjugated by $\Adr_{\exp(E(\fw))}$. As usual, the edge of
the wedge is invariant under such a conjugation, so we need
only determine the effects on the drift components of the system
as is done in Eqns.~\eqref{eqn:two-qubit-K123} through \eqref{eqn:two-qubit-P2}.
Moreover, the wedge is global by application of Corollary \ref{HHLglobalcor}.

\medskip
Now, the generalisation to systems with more than two qubits satisfying the (H)"~ or (WH)"~condition is obvious:
assuming uncorrelated noise, the $\fp$"~parts of the Lie wedges can be immediately
extended on the grounds of the previous description, since all processes are
local on each qubit. Though straightforward, calculating the $\fk$"~components becomes
a bit more tedious: but the many-body coherences have to be considered just as in \eqref{eqn:two-qubit-K123}.

\section{Outlook: Approximating Reachable Sets}
Knowing the {\em global} Lie wedge of a {\em coherently controlled}
Markovian system provides a convenient means
to efficiently approximate its reachable sets.
As in the case of a Lie algebra, the image of the wedge $\fw$ under
the exponential map yields a first approximation of the corresponding
Lie semigroup $\bS$. Unfortunately, this image is in general only
a proper subset of $\bS$---this, however, may happen also
for Lie algebras when the corresponding Lie group is non-compact.
Therefore, one has to allow for finite products of the form
${\rm e}^{A_1}{\rm e}^{A_2}\cdots{\rm e}^{A_\ell}$ with
$A_1,A_2,\dots, A_\ell \in \fw$ to obtain the entire semigroup $\bS$. 
Although the  minimal number $\ell_*$ of factors to generate $\bS$
(called \emph{number of intrinsic control-switches})
is in general unknown, this approach provides a much more effective
parametrization of the reachable sets than the standard method
which works with the original control directions and
piecewise constant controls as parameter space. Thereby one can optimize
target functions almost directly over the reachable sets thus
complementing standard optimal control methods of open systems
\cite{Teo91,JPB_decoh,Rabitz07}. Particularly simple are systems whose Lie wedges
do carry a Lie-semialgebra structure (like in isotropoic
depolarising channels). Here
one knows {\em a priori} that only a few (or sometimes even
zero) intrinsic control-switches are necessary, so some control
problems may actually be solved by {\em constant controls}.
\section{Conclusions}

We have generalised standard unital quantum channels (bit-flip,
phase-flip, bit-phase-flip, and depolarising) by allowing for different
degree of coherent Hamiltonian control. 
For the first time, here we have characterized their
respective global Lie wedges governing {\em all directions}
the controlled open system can possibly take.
The results have been further generalised
to various types of two-qubit systems with uncorrelated noise.
Since controlled multi-qubit channels can be treated likewise, 
the geometrical Lie-semigroup approach taken is anticipated
to find wide applications in quantum systems theory and engineering:
this is because knowing the global Lie wedge of a controlled
Markovian system paves the way to efficiently approximate its reachable sets.
Thus this 
knowledge  will be very useful for improving
known bounds (cf.~\cite{Yuan10}) on the corresponding system semigroup
$\bP_\Sigma$ in follow-up work.

Finally, our results demonstrate that the Lie wedges
associated to most of the controlled quantum systems
do \emph{not} take the special form of Lie semialgebras,
an important exception being the fully controlled isotropic
depolarising channel.

\section{Appendix}
\subsection{The Principal Theorem of Globality}\label{app:A}

For the reader's convenience, we state the `\/Principal Globality
Theorem\/' with minor simplifications. For the full version and
its (quite involved) proof we refer to \cite{HHL89}, which we sketch
in the sequel.

Let $\bG$ be a matrix Lie group with Lie algebra $\fg$, so
\begin{equation}
\fg \, X := \{AX \;|\; A \in \fg\}
\end{equation}
can  be envisaged as {\em tangent space} $T_X \bG $ at $X\in\bG$,
while $T\bG$ and $T^*\bG$ shall denote the {\em tangent bundle}
and, respectively, {\em cotangent bundle} of $\bG$. Thus, one
has the isomorphisms
\begin{equation}
T\bG \cong \fg \times \bG
\quad\text{and}\quad
T^*\bG \cong \fg^* \times \bG.
\end{equation}
%
%
%
Now, let $\fw$ be any wedge of $\fg$. A {\em 1-form} on $\bG$
is a smooth cross section of the cotangent bundle,
i.e.~$\omega: \bG\to T^*\bG$ with $\omega(X)\in T^*_X\bG$.
Moreover, $\omega$ is called
\begin{enumerate}
\item[(1)]
{\em exact} if there exists a smooth function
$\varphi: \bG \to \R{}$ such that ${\rm d}\varphi = \omega$;
\item[(2)] {\em $\fw$"~positive} at $X\in\bG$ if
$\langle\omega(X),AX\rangle \geq 0$ for all $A\in \fw$;
\item[(3)] {\em strictly $\fw$"~positive} at $X\in\bG$ if
$\fw$"~positivity holds at $X\in\bG$ and one has
$\expt{\omega(X),AX}>0$ for all  $A\in\fw\setminus -\fw$;
\end{enumerate}

\noindent
The existence of a strictly $\fw$"~positive 1"~form is ensured
in the following scenario \cite{HHL89}:
If $\bG$ is a Lie group with Lie algebra $\fg$ and
$\bH$ a closed subgroup with Lie algebra $\fh$,
then for any Lie wedge $\fw\subset\fg$ whose edge
$E(\fw)$ coincides with $\fh$ one can construct strictly
$\fw$"~positive 1"~forms on $\bG$.
Note, however, that these 1"~forms on $\bG$ are in general
\emph{not} exact. Yet, whenever exactness can be guaranteed
in addition, one has the following equivalences.

\medskip
\begin{theorem}[\cite{HHL89}]\label{thm:PGT}
Let $\bG$ denote a finite-dimensional real matrix Lie group with
Lie algebra $\fg$ and let $\fw$ be a Lie wedge of $\fg$.
Moreover, let $\mathfrak{g}_0:=\langle \fw \rangle_{\rm Lie}$ be the Lie
subalgebra generated by $\fw$ and let $\mathbf{G}_0$ be the corresponding
Lie subgroup of $\bG$. Further, assume that $\mathbf{G}_0$ is closed within
$\bG$. Then the following statements are equivalent:
\begin{enumerate}
\item[(a)]
$\fw$ is global in $\bG$.
\item[(b)]
$\fw$ is global in $\mathbf{G}_0$.
\item[(c)]
There is a closed connected subgroup $\bH$ of $\mathbf{G}_0$
with $L(\bH) = E(\fw)$ and a $1$-from $\omega$ on $\mathbf{G}_0$
which satisfies the following conditions:
\begin{enumerate}
\item[(i)]
$\omega$ is exact.
\item[(ii)]
$\omega$ is $\fw$-positive for all $X \in \mathbf{G}_0$.
\item[(iii)]
$\omega$ is strictly $\fw$-positive at the identity $\unity$.
\end{enumerate}
\end{enumerate}
\end{theorem}

\medskip
Now, the following consequence of the `\/Principal Globality Theorem\/'
already mentioned in the main text is a useful tool
whenever a {\em global} Lie wedge $\fw$ embracing the Lie wedge  $\fw_0$
of interest is already known.

\medskip
{\em Corollary~\ref{HHLglobalcor}~(\cite{HHL89})} 
Let $\bG$ be a Lie group with Lie algebra $\fg$ and let
$\fw_0\subseteq\fw$ be two Lie wedges in $\fg$.
Provided
\begin{eqnarray}
   \fw_0\setminus-\fw_0\subseteq\fw \setminus -\fw \;,
\end{eqnarray}
then $\fw_0$ is global in $\bG$ if the following conditions are
satisfied:
\begin{enumerate}
\item[(i)]
$\fw$ is global in $\bG$.
\item[(ii)]
The edge of $\fw_0$ 
is the Lie algebra of a closed Lie subgroup of $\bG$.
\end{enumerate}

\medskip
In other words, if the edge of the wedge follows the intersection
$E(\fw_0)=E(\fw)\cap\fw_0$
and $\fw$ is global, then $\fw_0$ is also a global Lie wedge,
whenever $\exp E(\fw_0)$ generates a closed subgroup.

\subsection{Lie Semialgebra Structure in Example 1}\label{app:B}

In Sec.~\ref{sec:geoR3} we stated that the Lie wedge of
Example~1
\begin{equation}
\fw_0 := \so(3)\;\oplus\; (-\fc_0),
\end{equation}
where $\fc_0 := \R{}^+_0\mathcal M(\Gamma_0)$ and
$M(\Gamma_0):=
\{S \in \sym(3) \,|\, S \prec \Gamma_0\}$,
is in fact a {\em Lie semialgebra} for $\Gamma_0 = \lambda \cdot \unity$
(corresponding to the isotropic depolarising channel),
whereas it fails to be a Lie semialgebra for any other $\Gamma_0$.
Recall, here $\sym(3)$ is the set of all symmetric
$3 \!\times\! 3$-matrices. For proving the above statement,
we distinguish the following
cases\footnote{Although case (iii) seems to be quite similar to
case (ii), its proof is more involved and a helpful preparation
of the general case (iv).}:
\begin{enumerate}
\item[(i)]
$\Gamma_0$ is a multiple of the identity, thus
we can assume without loss of generality $\Gamma_0 := \unity$;
\item[(ii)]
$\Gamma_0$ has zero as eigenvalue with multiplicity $2$, thus we
can assume $\Gamma_0 := \diag(1,0,0)$.
\item[(iii)]
$\Gamma_0$ has an eigenvalue different to zero with multiplicity $2$,
thus without loss of generality $\Gamma_0 :=\mbox{diag}(1,1,0)$.
\item[(iv)]
$\Gamma_0$ has three distinct eigenvalues, i.e.
$\Gamma_0 :=\mbox{diag}(a,b,c)$ with $a > b > c \geq 0$.
\end{enumerate}

In all cases, the identification of the dual wedge 
of $\fw_0$ is crucial for the compution of the
tangent space $T_A\fw_0$ at $A \in \fw_0$ via \eqref{eqn:semialg-incl}.
Therefore, we first provide an auxiliary result characterizing
the dual cone of $\fc_0$ within $\sym(3)$. 

\begin{lemma}
\label{dual}
Let $\Gamma_0 := \diag(a,b,c)$ with $a \geq b \geq c \geq0$ and let
$\fc_0 := \R{}^+_0\mathcal M(\Gamma_0)$ with
$ \mathcal{M}(\Gamma_0) :=
\{S \in \sym(3) \,|\, S \prec \Gamma_0\}$.
Then the dual cone of $\fc_0$ within $\sym(3)$ is given by
\begin{eqnarray}
\mathfrak{c}_{\mathfrak{p}}^*
:=\left\{S \in \sym(3) \,|\,
c \lambda_1(S) + b \lambda_2(S) + a \lambda_3(S) \geq 0\right\},
\end{eqnarray}
provided $\lambda_1(S) \geq \lambda_2(S) \geq \lambda_3(S)$
are the eigenvalues of $S$.
\end{lemma}

\begin{proof}
By definition, one has the equivalence $S \in \fc_{\fp}^*$ if
and only if $\langle S, S'\rangle\geq 0$ for all $S'\in\fc_0$. 
Since $\fc_0 = \R{}_0^+ \conv \mathcal{O}_{SO(3)}(\Gamma_0)$
this condition reduces to
$\langle S, \Theta \Gamma_0\Theta^\top \rangle \geq 0$ 
for all $\Theta \in SO(3)$. Then von Neumann's inequality
\cite{NEUM-37} provides the equivalence:
$\langle S, \Theta \Gamma_0\Theta^\top \rangle \geq 0$
for all $\Theta \in SO(3)$ if and only if 
$c \lambda_1(S) + b \lambda_2(S) + a \lambda_3(S) \geq 0 $,
where $\lambda_1(S) \geq \lambda_2(S) \geq \lambda_3(S)$
are the eigenvalues of $S$. Hence the result follows.
\end{proof}

Now, we are prepared to prove the above claim about the
Lie semialgebra property of $\fw_0$ 
\begin{proof} (i) In case $\Gamma_0=\unity$, the pointed cone
$\fc_0 := \R{}^+_0\mathcal M(\Gamma_0)$ equals the ray
$\R{}^+_0 \unity$. By Lemma \ref{dual}, we obtain
$\mathfrak{c}_{\mathfrak{p}}^* = \{S \in \sym(3) \,|\, \tr S \geq 0\}$
and thus
$\fc_0^* = \mathfrak{sl}(3,\mathbb R) \oplus \mathbb R^+_0 \Gamma_0$,
where $\mathfrak{sl}(3,\R{})$ denotes the set of all
$3 \!\times\! 3$-matrices with trace zero.
Hence
\begin{equation}
\fw^*_0 = \so(3)^\perp \cap \fc^*_0 = \{S \in \sym(3)\;|\; \tr S \geq 0\}\;.
\end{equation}
For $A := \lambda \unity + \Omega \in \fw_0$ with $\lambda \geq 0$
and $\Omega \in \so(3)$ it follows
\begin{equation}
A^\perp \cap \fw^*_0 =
\begin{cases}
\{S \in \sym(3)\;|\; \tr S = 0\} & \text{for } \lambda > 0,\\[2mm]
\{S \in \sym(3)\;|\; \tr S \geq 0\} & \text{for } \lambda = 0,
\end{cases}
\end{equation}
and thus
\begin{equation}
T_A \fw_0  = (A^\perp \cap \fw^*_0)^\perp =
\begin{cases}
\so(3) \oplus \R{} \unity & \text{for } \lambda > 0,\\[2mm]
\so(3) & \text{for } \lambda = 0.
\end{cases}
\end{equation}
Thereby the inclusion
$[A, T_A \fw_0] \subset T_A \fw_0$ is obviously always satisfied
and hence $\fw_0$ {\em is a Lie semialgebra for
$\Gamma_0=\unity$}.

\medskip
\noindent
(ii) In case $\Gamma_0=\diag(1,0,0)$, it is easy to see that the
pointed cone $\fc_0 := \R{}^+_0\mathcal M(\Gamma_0)$ actually
consists of all positive semidefinite
$3 \!\times\! 3$-matrices. Here, Lemma \ref{dual} yields
$\mathfrak{c}_{\mathfrak{p}}^* = \fc_0$. This reflects the
well-known fact that the cone of all positive semidefinite
matrices is self-dual \emph{within the space of all symmetric
matrices}. Hence
\begin{equation}
\fw^*_0 = \so(3)^\perp \cap \fc^*_0 = \sym(3) \cap \fc^*_0 = \fc_0\;.
\end{equation}
Now, for
$
A :=  \Gamma_0 + H_z =
\left[\begin{smallmatrix}
1 & 0 & 0\\
0 & 0 & 0\\
0 & 0 & 0
\end{smallmatrix}\right]
+
\left[\begin{smallmatrix}
0 & -1 & 0\\
1 & 0 & 0\\
0 & 0 & 0
\end{smallmatrix}\right]
\in \fw_0
$
we obtain
\begin{equation}
A^\perp \cap \fw^*_0 =
\left\{
\begin{bmatrix}
0 & 0\\
0 & S
\end{bmatrix}
\;\Big|\;
S\; \in \sym(2)\,,\, S\geq 0 \right\}
\end{equation}
and therefore
\begin{equation}
T_A \fw_0  = (A^\perp \cap \fw^*_0)^\perp
 =
\so(3) \oplus
\text{span}
\left\{
\Gamma_0,
p_y,
p_z
\right\}.
\end{equation}
Finally, for disproving the inclusion
$[A, T_A \fw_0] \subset T_A \fw_0$ consider the
commutator of $A\in\fw_0$ and
$
B := p_z
\in\;T_A \fw_0\;.
$
It follows
$
[A,B] = 
-H_z
+
\diag(-2,2,0)
$
which clearly violates the inclusion $[A, T_A \fw_0] \subset T_A \fw_0$.
Thus $\fw_0$ is {\em not a Lie semialgebra for $\Gamma_0=\diag(1,0,0)$}.

\medskip
\noindent
(iii) In case $\Gamma_0=\diag(1,1,0)$, we obtain by Lemma
\ref{dual} the following description
\begin{equation*}
\fw^*_0 = \so(3)^\perp \cap \fc^*_0 = \fc^*_{\fp}=
\{S \in \sym(3) \,|\, \lambda_2(S) +\lambda_3(S) \geq 0\}\;.
\end{equation*}
Now, let $A:=\Gamma_0+H_y$. Then, it is easy to see that
\begin{equation*}
\begin{split}
A^\perp &\cap \fw^*_0 \supseteq \\
& \R{}^+_0\text{conv}
\left\{
\diag(1,-1,1),
\diag(-1,1,1),
(p_z+E_{33})
\right\}\,.
\end{split}
\end{equation*}
Moreover,
for $S \in A^\perp \cap \fw^*_0$ one has the conditions
\begin{equation*}
\langle\Gamma_0,S\rangle = 0
\quad\text{and}\quad
\langle\Theta\Gamma_0\Theta^\top,S\rangle \geq 0
\end{equation*}
for all $\Theta \in SO(3)$. 
Now, differentiating the second condition with respect to
$\Theta \in SO(3)$ shows $\langle[\Omega,\Gamma_0],S\rangle = 0$
for all $\Omega \in \so(3)$, i.e.~$[\Omega,\Gamma_0]$ 
belongs to $(A^\perp \cap \fw^*_0)^\perp$. Thus one has
\begin{equation*}
T_A \fw_0   = (A^\perp \cap \fw^*_0)^\perp
 \supseteq \so(3) \oplus \R{}A \oplus
\text{span}
\left\{
p_x,
p_y
\right\}.
\end{equation*}
Hence, counting dimensions finally yields
\begin{equation*}
T_A \fw_0  = \so(3) \oplus
\text{span}
\left\{
\Gamma_0,
p_x,
p_y
\right\}\,.
\end{equation*}
To disprove the set inclusion $[A, T_A \fw_0] \subset T_A \fw_0$
consider the commutator of $A\in\fw_0$ and
$
B := 
p_y
\in\;T_A \fw_0\;.
$
The computation is left to the reader (see Tab.~\ref{tab:comm-tab}). The
result clearly violates the inclusion $[A, T_A \fw_0] \subset T_A \fw_0$
and thus $\fw_0$ is {\em not a Lie semialgebra} for
$\Gamma_0=\diag(1,1,0)$ either.

\medskip
\noindent
(iv) For $\Gamma_0=\diag(a,b,c)$ with $a > b > c \geq 0$ and
$A := \Gamma_0 + H_\nu$ with $\nu \in \{x,y,z\}$, the same
arguments as above show that $T_A \fw_0$ is given by
\begin{equation*}
T_A \fw_0  = \so(3) \oplus
\text{span}
\left\{
\Gamma_0,
p_x,
p_y,
p_z
\right\}\,.
\end{equation*}
Therefore, an appropriate choice of $B=p_\nu$ with
$\nu \in \{x,y,z\}$ demostrates again that $\fw_0$ is
{\em not a Lie semialgebra} in the general case
$\Gamma_0=\diag(a,b,c)$ with $a > b > c \geq 0$ either.
\end{proof}

Note that in all the above cases the tangent space
of $\fw_0$ has the following form
\begin{equation*}
T_A \fw_0  = \so(3) \oplus \R{}\Gamma_0 \oplus
\rT_{\Gamma_0} \mathcal{O}_{SO(3)}(\Gamma_0)\,,
\end{equation*}
where the tangent space 
of the orbit $\mathcal{O}_{SO(3)}(\Gamma_0)$
at $\Gamma_0$ is given by 
$\rT_{\Gamma_0} \mathcal{O}_{SO(3)}(\Gamma_0)
= \{[\Omega,\Gamma_0] \,|\, \Omega \in \so(3)\}$.

\bibliographystyle{IEEEtran}
\bibliography{IEEEcontrol21}
\end{document}